\documentclass[%
 reprint,
 superscriptaddress,
 amsmath,amssymb,
 aps,
pra,
]{revtex4-2}

\usepackage[utf8]{inputenc} 
\usepackage{amsmath,amsfonts,amssymb,mathrsfs,amsthm,cases} 
\usepackage{bm, bbm, dsfont} 
\usepackage{physics} 
\usepackage{graphicx} 
\usepackage{diagbox,threeparttable,dashrule,booktabs,dcolumn} 
\usepackage{xcolor} 
\usepackage{url} 
\usepackage{hyperref} 
\hypersetup{colorlinks,linkcolor={blue},citecolor={blue},urlcolor={red}}

\usepackage{cleveref} 

\crefname{equation}{Eq.}{Eqs.}
\crefname{section}{Section}{Sections}
\crefname{definition}{Definition}{Definitions}
\crefname{proposition}{Proposition}{Propositions}
\crefname{lemma}{Lemma}{Lemmas}
\crefname{theorem}{Theorem}{Theorems}
\crefname{corollary}{Corollary}{Corollaries}
\crefname{conjecture}{Conjecture}{Conjectures}
\crefname{claim}{}{Claims}
\crefname{example}{Example}{Examples}

\newtheorem{definition}{Definition}
\newtheorem{proposition}{Proposition}
\newtheorem{lemma}{Lemma}

\newtheorem{theorem}{Theorem}

\newcommand{\proj}[1]{|{#1}\rangle \langle {#1}|}

\newcommand{\I}{\mathds{I}}

\def\bpf{\begin{proof}}
\def\epf{\end{proof}}
\def\bea{\begin{eqnarray}}
\def\eea{\end{eqnarray}}
\def\beq{\begin{equation}}
\def\eeq{\end{equation}}
\def\bal{\begin{aligned}}
\def\eal{\end{aligned}}
\def\bma{\begin{pmatrix}}
\def\ema{\end{pmatrix}}

\def\tr{{\rm Tr}}

\def\dg{\dagger}

\def\ox{\otimes}

\def\lin{\mathop{\rm span}}
\def\supp{\mathop{\rm Supp}}
\def\inv{\mathop{\rm inv}}

\def\a{\alpha}
\def\b{\beta}

\def\e{\epsilon}

\def\t{\theta}

\def\la{\lambda}

\def\x{\xi}

\def\om{\omega}

\newcommand{\nc}{\newcommand}

\nc{\bbA}{\mathbb{A}} \nc{\bbB}{\mathbb{B}} \nc{\bbC}{\mathbb{C}}
\nc{\bbD}{\mathbb{D}} \nc{\bbE}{\mathbb{E}} \nc{\bbF}{\mathbb{F}}
\nc{\bbG}{\mathbb{G}} \nc{\bbH}{\mathbb{H}} \nc{\bbI}{\mathbb{I}}
\nc{\bbJ}{\mathbb{J}} \nc{\bbK}{\mathbb{K}} \nc{\bbL}{\mathbb{L}}
\nc{\bbM}{\mathbb{M}} \nc{\bbN}{\mathbb{N}} \nc{\bbO}{\mathbb{O}}
\nc{\bbP}{\mathbb{P}} \nc{\bbQ}{\mathbb{Q}} \nc{\bbR}{\mathbb{R}}
\nc{\bbS}{\mathbb{S}} \nc{\bbT}{\mathbb{T}} \nc{\bbU}{\mathbb{U}}
\nc{\bbV}{\mathbb{V}} \nc{\bbW}{\mathbb{W}} \nc{\bbX}{\mathbb{X}}
\nc{\bbZ}{\mathbb{Z}}

\nc{\cB}{{\cal B}} \nc{\cH}{{\cal H}} \nc{\cE}{{\cal E}} \nc{\cS}{{\cal S}}
\nc{\cG}{{\cal G}} \nc{\cZ}{{\cal Z}}

\nc{\bo}{{\bf o}}

\nc{\Dp}{\Delta p} \nc{\Dq}{\Delta q}

\def\etal{{\sl et~al.~}}

\begin{document}


\title{Quantum hypothesis testing between qubit states with parity} 

\author{Yi Shen}
\email[]{yishen@jiangnan.edu.cn}
\affiliation{School of Science, Jiangnan University, Wuxi, Jiangsu 214122, China}

\author{Carlo Maria Scandolo}
\email[]{carlomaria.scandolo@ucalgary.ca}
\affiliation{Department of Mathematics and Statistics, University of Calgary, Calgary, AB T2N 1N4, Canada}
\affiliation{Institute for Quantum Science and Technology, University of Calgary, Calgary, AB T2N 1N4, Canada}

\author{Lin Chen}
\email[]{linchen@buaa.edu.cn (corresponding author)}
\affiliation{School of Mathematical Sciences, Beihang University, Beijing 100191, China}
\affiliation{International Research Institute for Multidisciplinary Science, Beihang University, Beijing 100191, China}

\date{\today} 

\begin{abstract}
Quantum hypothesis testing (QHT) provides an effective method to discriminate between two quantum states using a two-outcome positive operator-valued measure (POVM). Two types of decision errors in a QHT can occur. In this paper we focus on the asymmetric setting of QHT, where the two types of decision errors are treated unequally, considering the operational limitations arising from the lack of a reference frame for chirality. This reference frame is associated with the group $\bbZ_2$ consisting of the identity transformation and the parity transformation. Thus, we have to discriminate between two qubit states by performing the $\bbZ_2$-invariant POVMs only. 
We start from the discrimination between two pure states. By solving the specific optimization problem we completely characterize the asymptotic behavior of the minimal probability of type-II error which occurs when the null hypothesis is accepted when it is false. Our results reveal that the minimal probability reduces to zero in a finite number of copies, if the $\bbZ_2$-twirlings of such two pure states are different. We further derive the critical number of copies such that the minimal probability reduces to zero. Finally, we replace one of the two pure states with a maximally mixed state, and similarly characterize the asymptotic behavior of the minimal probability of type-II error.
\end{abstract}


\maketitle


\section{Introduction}
\label{sec:intro}

Hypothesis testing refers to the fundamental procedures for statisticians to accept or reject statistical hypotheses based on observed data taken by a collection of random variables \cite{kendall1994}.
There are two hypotheses in a statistical hypothesis testing. 
One is the null hypothesis, denoted by $H_0$, which assumes that individuals in a population are randomly distributed among the sampling units of a sample. The other one is the alternative hypothesis, denoted by $H_1$, which is opposite to the null hypothesis \cite{2005testing}.
Due to the limited sample data, there would be two types of decision errors resulting from a hypothesis testing. The type-I error occurs when the null hypothesis is rejected when it is true. The type-II error occurs when the null hypothesis is accepted when it is false. The ultimate purpose of hypothesis testing is to formulate an optimal strategy to make a decision with the minimal error probability typically of type-II errors.

Quantum hypothesis testing (QHT) \cite{hp1991} is the counterpart of statistical hypothesis testing in quantum information theory. It provides an effective method to discriminate between two quantum states using a two-outcome positive operator-valued measure (POVM). Assume that a given quantum source prepares $n$ independent copies of a quantum system in either state $\rho_0$ or $\rho_1$. We assign the null hypothesis to $\rho_0^{\ox n}$ and the alternative hypothesis to $\rho_1^{\ox n}$, and perform the binary POVM $\{E_n,\I-E_n\}$ to test the $n$-copy quantum system. If the measurement outcome is associated with $E_n$ or $\I-E_n$, then we respectively determine that the system is prepared in $\rho_0^{\ox n}$ or $\rho_1^{\ox n}$. Similar to the statistical hypothesis testing, two types of decision errors in a QHT would occur \cite{erate2008} as follows.
\begin{itemize}
\item Type-I error: The observer determines that the state is $\rho_1^{\ox n}$, while in reality it is $\rho_0^{\ox n}$.
\item Type-II error: The observer determines that the state is $\rho_0^{\ox n}$, while in reality it is $\rho_1^{\ox n}$.
\end{itemize} 
The type-I error happens with the probability $\a_n:=\tr[\rho_0^{\ox n}(\I-E_n)]$ and the type-II error happens with the probability $\b_n:=\tr[\rho_1^{\ox n}E_n]$.
The discrimination task is to decide which hypothesis is true based on the data drawn from an optimal POVM which leads to the minimal error probability, as detailed below. 

The setting of QHT can be classified as symmetric and asymmetric depending on whether the two types of decision errors are treated equally. In the symmetric setting, the two types of errors are treated equally, where the purpose is to minimize the average (Bayesian) of two error probabilities weighted by the prior probabilities of generating $\rho_0$ and $\rho_1$ \cite{erate2008}.  
That is to find the optimal POVM which minimizes
\beq
\label{eq:symep}
P_{e,n}:=\pi_0\a_n+\pi_1\b_n,
\eeq
where $\pi_0$ and $\pi_1$ denote the prior probabilities with $\pi_0+\pi_1=1$. One can verify that 
\beq
\label{eq:symep-1}
\min\limits_{0\leq E_n\leq I} P_{e,n}=\frac{1}{2}(1-\norm{\pi_0\rho_0^{\ox n}-\pi_1\rho_1^{\ox n}}_1),
\eeq
where $\norm{\cdot}_1$ represents the trace norm.
Further, it has been shown that the asymptotic rate of minimal $P_{e,n}$ is completely characterized by the quantum Chernoff bound \cite{qchernoff2007}. Chernoff \cite{chernoff1952} identified that the minimal average error probability in discriminating two probability distributions decreases exponentially in the number of tests $n$, and the optimal exponent arising in the asymptotic limit is known as the celebrated Chernoff bound. For the minimal average error probability $P_{e,n}$ arising from the QHT, Audenaert \etal~\cite{qchernoff2007} finally figured out its asymptotic rate. That is, $P_{e,n}\sim \exp(-n\x_{QCB})$, where $\x_{QCB}=-\log\min\limits_{0\leq s\leq 1}\tr(\rho_0^s\rho_1^{1-s})$ known as the quantum Chernoff bound. To be consistent with the expressions in Ref. \cite{qchernoff2007} we shall henceforth omit the logarithmic base $2$, unless the bases require to be specified.
In this paper, we are more interested in the asymmetric setting QHT, where the two types of errors cannot be treated equally.
As a result, the purpose of asymmetric QHT is to minimize the probability of type-II error under the condition when the probability of type-I error is bounded by a constant $\e$. In Ref. \cite{hp1991}, Hiai and Petz discovered that for each $\e>0$ there exists a POVM such that the probability of type-II error decreases exponentially in the number of copies $n$: $\beta_n(\e)\sim \exp(-nr)$ with $r\geq D(\rho_0||\rho_1)$ where $D(\rho_0||\rho_1)$ is the quantum relative entropy (or quantum Kullback-Leibler divergence). Ogawa and Nagaoka \cite{qsl2000} further improved the above result and proved that the optimal exponent $r$ arising in the asymptotic limit is exactly the quantum relative entropy. This is the well-known quantum Stein's lemma. Since then the relation between the optimal error exponent and the quantum relative entropy has been further studied, and the quantum Stein's lemma has been generalized to many different situations \cite{Hayashi_2002,Hayashi_2007,erate2008,qstein2010,like2014}.

The above-mentioned are ideal results. Here, we consider the QHT in practical scenarios, where the observers cannot perform all POVMs due to the limitation of a practical setup. It leads to imposing corresponding restrictions on the POVMs. As we know, every restriction on quantum operations defines a quantum resource theory (QRT) by partitioning all states into two groups, one consisting of free states and the other consisting of resourceful states \cite{qrt2019}. Accompanying the set of free states is a collection of free quantum operations leaving the set of free states invariant, and the QRT studies what information processing tasks become possible using the free operations. Based on the QRT and from a practical point of view, it is interesting to ask how the minimal error probability behaves by only performing free POVMs. For example, the QHT with the restriction of local operations and classical communication (LOCC) has been studied in Refs. \cite{restrictedqht2017,restrictedqht2020}. In this paper, we are interested in the restriction of a superselection rule (SSR) \cite{rfgour2008} which implies a kind of symmetry.

Symmetry plays a central role in physics. Making full use of symmetry can simplify the study of properties and evolutions of physical systems \cite{asym2016}. Consequently, finding the consequences of symmetries for dynamics is a subject of broad applications in physics \cite{Marvianphd,marviannc2014}. The restriction of a SSR arises from the lack of some quantum reference frame \cite{ssr2007,rfgour2008}, when the state one possesses was prepared according to some reference frame one does not have access to. It leads to the resource theory of asymmetry \cite{Marvianphd}, where the quantum reference frame is associated with a group of transformations, denoted by $G$. A symmetric evolution is one that commutes with the action of the symmetry group \cite{ssr2007}. Then in the situation of lacking some quantum reference frame, the free states (operations) are those $G$-invariant states (operations) that are invariant under the action of group $G$ \cite{rfgour2008}. Such free states and operations are interchangeably called $G$-symmetric states and operations respectively. It is worth noting that the authors in Ref. \cite{qhtgs2009} studied the asymptotic discrimination problem of two quantum states with $G$-invariant measurements, and derived bounds on various asymptotic error exponents. Specifically, we consider the group $\bbZ_2$ consisting of two transformations, one is the identity transformation and the other is the parity transformation. The parity transformation which flips the sign of one spatial coordinate is an important concept in quantum mechanics. The group $\bbZ_2$ is associated with a reference frame for chirality. Such a frame is the component of a Cartesian frame with respect to which the handedness of a quantum system is defined \cite{rfgour2008}. One may refer to \cite[Sec. 4]{rfgour2008} for more details about chiral frames and the resource theory of the $\bbZ_2$-SSR. Due to the lack of a chiral reference frame, the measurements are specifically required to be $\bbZ_2$-invariant ones.


\begin{figure}[ht]
\centering
\includegraphics[width=0.48\textwidth]{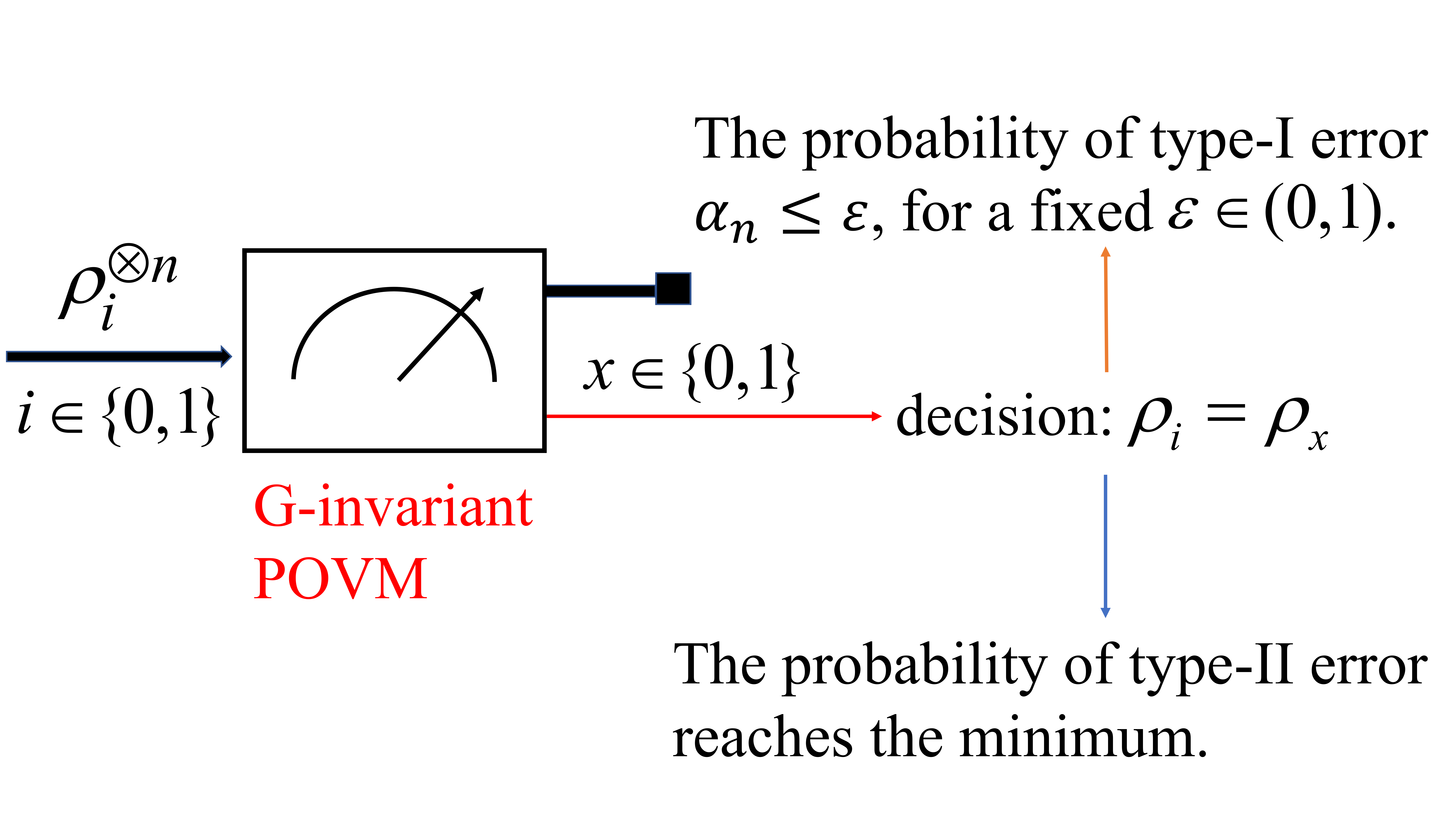}
\caption{The asymmetric QHT within the resource theory of asymmetry associated with some group $G$. The performed POVMs have to be $G$-invariant. The main task of this process is to minimize the probability $\b_n$ when the probability $\a_n$ is smaller than a constant $\epsilon$. In this paper, we specify the group $G$ as $\bbZ_2$ arising from the lack of a chiral reference frame and focus on two pure qubit states $\rho_0$ and $\rho_1$.}
\label{fig:asyqht}
\end{figure}

Now we are able to introduce the problem we focus on in this paper. It can be summarized as investigating the asymmetric QHT between two qubit states by only performing the restricted POVMs arising from the resource theory of parity. We also illustrate this specific problem by Fig. \ref{fig:asyqht}, where the probability of type-I error is tolerated within a small constant and the second probability has to be minimal. The investigation is completed by the following steps. First, we introduce an equivalent way to characterize $G$-invariant states (operations) by the so-called $G$-twirling operation for an arbitrary group $G$. The advantage of this characterization shows that every free POVM, i.e. $G$-invariant POVM, is indeed the $G$-twirling of some POVM. 
Second, we assign ourselves a specific group $\bbZ_2$ for the lack of a chiral reference frame, and formulate the analytic expression of the $\bbZ_2$-twirling acting on an arbitrary pure qubit state for any copies. Based on the explicit expressions, we obtain the accurate value of the minimal probability of type-II error when the number of copies is large enough. Note that the accurate value implies the exact behavior of the minimal probability of type-II error, rather than the bound on the asymptotic error exponent correspondingly studied in Ref. \cite{qhtgs2009}. Then we figure out the specific problem for two pure qubit states by Theorem \ref{th:main}. We also illustrate the main results in Theorem \ref{th:main} by Table \ref{tab:pure}. In Theorem \ref{th:main} (i) we discover that there exist two distinct pure qubit states whose $\bbZ_2$-twirlings are the same for any copies. Thus, in this case the minimal probability of type-II error is a non-zero constant $1-\epsilon$ as the number of copies increases. In Theorem \ref{th:main} (ii) we show that if two pure qubit states have different $\bbZ_2$-twirlings, then there exists a finite number of copies $n_\e$ related to the parameter $\e$ such that the minimal probability of type-II error reduces to zero for any $n\geq n_\e$. It implies that we can perfectly decide that the given system is prepared in the state $\rho_0^{\ox n}$ (the null hypothesis) using $\bbZ_2$-invariant POVMs for only a finite number of copies by definition. Theorem \ref{th:main} (ii) consists of three parts which are supported by Lemmas \ref{pp:generaltwoqubit} - \ref{pp:reduced-2} respectively. Furthermore, it is important to derive the critical number of copies $n_\e$ from the perspective of resource conservation as fewer copies require less resource to prepare. By fixing the optimal $\bbZ_2$-invariant  POVM, we derive the critical number $n_\e$ and estimate how large it is as $\e\to 0$ in Proposition \ref{pp:ne}. Finally, we extend our study to the context of the QHT between a pure state and a maximally mixed state. By replacing one of the two pure states with a maximally mixed state, we characterize the asymptotic behavior of the minimal probability of type-II error in Theorem \ref{thm:pandm}. The main results in Theorem \ref{thm:pandm} are also presented in Table \ref{tab:p+mm} for convenience. When $\rho_1$ (the alternative hypothesis) is a maximally mixed state, it follows from Theorem \ref{thm:pandm} (i) that the minimal probability of type-II error is nonzero and decreases in the number of copies. When $\rho_0$ (the null hypothesis) is a maximally mixed state, we similarly conclude from Theorem \ref{thm:pandm} (ii) that there exists some $\bbZ_2$-invariant POVM such that the probability of type-II error reduces to zero in a finite number of copies. In this case, the critical number of copies for a specific optimal POVM is also obtained in Theorem \ref{thm:pandm} (ii).

The remainder of this paper is organized as follows. In Sec. \ref{sec:pre} we first formulate the mathematical setting of general asymmetric QHT, and then introduce the resource theory of asymmetry by formally defining the $G$-invariant states and operations. In Sec. \ref{sec:z2ssr} we specifically investigate how the hypothesis testing relative entropy between two pure qubit states asymptotically behaves within the resource theory of asymmetry associated with the group $\bbZ_2$. First, we clarify the problem in Sec. \ref{subsec:prob}. Second, we completely solve this problem in Sec. \ref{subsec:z2qht2qubit}. We further derive the critical number of copies when the minimal probability of type-II error reduces to zero in Sec. \ref{subsec:critical}. In Sec. \ref{sec:mixedhyre} we extend our study to the context of the QHT between a pure state and a maximally mixed state. The concluding remarks are given in Sec. \ref{sec:con}. Finally, we provide the detailed proofs of several crucial results in the appendices.


\section{Preliminaries}
\label{sec:pre}

In this section we introduce the preliminaries about the asymmetric QHT and the resource theory of asymmetry as the background to describe our problem.
In Sec. \ref{subsec:aqht} we formulate the mathematical setting of general asymmetric QHT and present the related results on the general asymmetric QHT.
In Sec. \ref{subsec:qrta} we formally define the $G$-invariant states (operations), and introduce an equivalent way to characterize them via the so-called $G$-twirling operation.

\subsection{Mathematical Setting of Asymmetric QHT}
\label{subsec:aqht}

Our investigation is based on the framework of asymmetric QHT, and the key difference from the general asymmetric QHT is that the POVMs performed in our investigation have to be restricted due to some practical limitations, for example, the lack of some quantum reference frames, see Fig. \ref{fig:asyqht}. Hence, it is necessary to first formulate the framework of asymmetric QHT.
Suppose that the two quantum hypotheses are $H_0:\rho_0^{\ox n}$ (null) and $H_1:\rho_1^{\ox n}$ (alternative), where $\rho_0,\rho_1\in\cB(\cH)$, and $\rho_0^{\ox n},\rho_1^{\ox n}\in\cB(\cH^{\ox n})$. Physically, the decision between the two hypotheses depends on the result after performing a binary POVM measurement. 
We first identify the two types of error probabilities as follows.

\begin{definition}
\label{def:errorprob}

For any positive integer $n$, suppose that $\{E_n,\I-E_n\}$ is a binary POVM selected, where $0\leq E_n\leq \I$ acting on $\cH^{\ox n}$. The type-I error occurs when $H_1$ is accepted by mistake, and the type-II error occurs when $H_0$ is accepted by mistake. Then their probabilities are given by
\begin{enumerate}
 \item Probability of type-I error: $\a_n\equiv\tr\big(\rho_0^{\ox n}(\I-E_n)\big)=1-\tr(\rho_0^{\ox n} E_n)$,
 \item Probability of type-II error: $\b_n\equiv\tr(\rho_1^{\ox n}E_n)$.
\end{enumerate}
\end{definition}

In hypothesis testing, the costs associated to the two types of error can be widely different, or even incommensurate \cite{erate2008}, which leads to the asymmetric setting of hypothesis testing. Analogously, in the asymmetric QHT we treat the two types of error formulated in Def. \ref{def:errorprob} unequally. Thus, the purpose is to minimize the probability of type-II error, when the probability of type-I error can be tolerated within a small $\e$. Specifically, we shall consider the following optimization problem:
\beq
\label{eq:asymm-1}
\b_n(\e):=\inf_{0\leq E_n\leq \I}\{\b_n:\a_n\leq \e\}.
\eeq
Note that this optimization problem is a semidefinite program (SDP) \cite{geentropy2012} and can be efficiently solved when the number of copies $n$ is not very large. Furthermore, using the duality of SDP, it is possible to write $\b_n(\e)$ as the maximum of a simple function of one real variable \cite{lorenz}, namely $\b_n(\e)\equiv\max_{r\geq 0} f_{n,\e}(r)$, where
\beq
\label{eq:dualsdp}
\bal
f_{n,\e}(r)&:=(1-\e)r-\tr(r\rho_0^{\ox n}-\rho_1^{\ox n})_+ \\
       &=\frac{1}{2}\left[1+(1-2\e)r-\norm{r\rho_0^{\ox n}-\rho_1^{\ox n}}_1\right].
\eal
\eeq
This observation was proposed in Ref. \cite{lorenz} to considerably simplify the analysis in formulating the framework of quantum relative Lorenz curves. Based on this transformation $\b_n(\e)\equiv\max_{r\geq 0} f_{n,\e}(r)$, the original optimization problem Eq. \eqref{eq:asymm-1} becomes more practical to process, due to being a function of a real variable.

It is important to derive the convergence rate of $\b_n(\e)$ as the number of copies, $n$, approaches infinity. It leads to the characterization, known as the hypothesis testing relative entropy, defined by
\beq
\label{eq:qhtre}
D_H^{\e}(\rho_0^{\ox n}||\rho_1^{\ox n}):=-\log \b_n(\e), ~\forall n.
\eeq 
It is a generalized relative entropy related to QHT. It is worth noting that $D_H^{\e}(\cdot||\cdot)$ satisfies the data processing inequality \cite{qhtligong}. That is for any two states $\rho,\sigma$ and $\e\in[0,1)$ the following inequality holds.
\beq
\label{eq:dpiqhtre}
D_H^{\e}(\rho||\sigma)\geq D_H^{\e}\big(\cE(\rho)||\cE(\sigma)\big),
\eeq
where $\cE$ is a completely positive map.

In information theory, the relative entropy is a fundamental quantity to compare two different probability distributions. It was first introduced by Kullback and Leibler, and thus the relative entropy is also called Kullback–Leibler divergence \cite{kullback1951}. Then it was generalized to the quantum relative entropy by Umegaki \cite{umegaki1962}. For any two states $\rho$ and $\sigma$, the quantum relative entropy between them is defined by
\beq
\label{eq:qredef}
D(\rho||\sigma):=
\left\{
\bal
&\tr\big(\rho(\log\rho-\log\sigma)\big),&~ \text{if $\supp\rho\subseteq\supp\sigma$,}\\
&+\infty,&~\text{otherwise,}
\eal
\right.
\eeq
where $\supp\rho$ denotes the support projection of state $\rho$.
The quantum relative entropy has been discovered to play an essential role in various aspects of quantum information theory \cite{reereview2002}. In addition to the hypothesis testing relative entropy, there are several other generalized relative entropies, such as the information spectrum relative entropy \cite{spectrumre2002} and the min- and max-relative entropies \cite{minmaxre2008}. These quantities are related to each other, and have been proved to be of important operational significance in both classical and quantum information theory to obtain the optimal rates of protocols \cite{minmaxre2008,encoding2019}.

In the following we present the celebrated result on the general asymmetric QHT. It reveals the essential connection between the hypothesis testing relative entropy and the quantum relative entropy.
\begin{lemma}[Quantum Stein’s Lemma]\cite{qsl2000}
\label{le:qstein}
For any $0<\e<1$, we have
\beq
\label{eq:qsteine}
\lim_{n\to \infty}\frac{1}{n} D_H^{\e}(\rho_0^{\ox n}||\rho_1^{\ox n})=D(\rho_0||\rho_1),
\eeq
where $D(\rho_0||\rho_1)$ is the quantum relative entropy.
\end{lemma}
From a practical point of view, Lemma \ref{le:qstein} states that for each $0<\e<1$ there exists a POVM such that the probability of type-II error decreases exponentially in the number of copies $n$: $\beta_n(\e)\sim \exp(-nr)$ with the optimal exponent $r=D(\rho_0||\rho_1)$.

\subsection{The Resource Theory of Asymmetry}
\label{subsec:qrta}

Denote by $G$ the group of transformations associated with the reference frame of concern. If we face the restriction of lacking this reference frame, the free states are those that can be prepared without access to the frame, and thus are invariant under these transformations in the group $G$. 
Denote by $\cB(\cH)$ the bounded linear operators on the given Hilbert space $\cH$. 
Suppose that $U:~G\to\cB(\cH)$ is the unitary representation of group $G$ that corresponds to the physical transformations in $G$. Then the free states are characterized by the $G$-invariant (symmetric) states satisfying
\beq
\label{eq:ginv-1}
U(g)\rho U^\dg(g)=\rho,\quad \forall g\in G.
\eeq
Let $\cS(\cH)$ be the set of normalized states. The set of $G$-invariant states is denoted by $\inv(G)$,
\beq
\label{eq:ginv-2}
\inv(G)\equiv\{\rho|~\forall g\in G: U(g)\rho U^\dg(g)=\rho,\rho\in\cS(\cH)\}.
\eeq
Similarly, the free operations are characterized by the $G$-invariant (symmetric) operations satisfying 
\beq
\label{eq:ginv-3}
U(g)\circ\cE \circ U^\dg(g)=\cE,\quad \forall g\in G,
\eeq
where $\cE$ is a completely positive map.

There are two equivalent ways to characterize the set $\inv(G)$ by the useful $G$-twirling operation \cite{gour2009}.
Let $\cG:~\cB(\cH)\to\cB(\cH)$ be the trace-preserving completely positive linear map defined by
\beq
\label{eq:gtwirling}
\cG[\rho]\equiv\int_G U(g)\rho U^\dg(g)~dg.
\eeq
which averages over the action of the group $G$ with the $G$-invariant (Haar) measure $dg$. If $G$ is finite, one simply replaces the integral with a sum as
\beq
\label{eq:gtwirling-1}
\cG[\rho]\equiv\frac{1}{\abs{G}}\sum_{g\in G} U(g)\rho U^\dg(g).
\eeq
Here, $\cG$ is known as the $G$-twirling operation. One can verify that the $G$-twirling operation has the following two basic properties. First, for a linear operator $\rho\in\cB(\cH)$, $\cG[\rho]=\rho$ if and only if $U(g)\rho U^\dg(g)=\rho,~\forall g\in G$. Thus, the set $\inv(G)$ is equivalent to the set of fixed points of $\cG$
\beq
\label{eq:gtwirling-2}
\inv(G)=\{\rho|\cG[\rho]=\rho,\rho\in\cS(\cH)\}.
\eeq
Second, $\cG$ is idempotent, namely $\cG^2\equiv\cG\circ\cG=\cG$. Thus, the set $\inv(G)$ is equivalent to the image of $\cG$ \cite{gour2009} 
\beq
\label{eq:gtwirling-3}
\inv(G)=\{\rho|\rho=\cG[\sigma],~\forall\sigma\in\cS(\cH)\}.
\eeq

If the operator $\rho$ in Eq. \eqref{eq:gtwirling} or Eq. \eqref{eq:gtwirling-1} is replaced with a quantum operation $\cE$, we can similarly obtain the $G$-twirling of $\cE$ as
\beq
\label{eq:gtwirling-4}
\cG[\cE]=\int_G U(g)\circ\cE\circ U^\dg(g)~dg.
\eeq
From Eq. \eqref{eq:gtwirling-3} we also conclude that any $G$-invariant operation can be expressed by $\cG[\cE]$ for some completely positive map $\cE$.

\section{Quantum hypothesis testing between two pure qubit states with parity}
\label{sec:z2ssr}

In this section we specifically investigate how the hypothesis testing relative entropy between two pure qubit states asymptotically behaves within the resource theory of asymmetry associated with the group $\mathbb{Z}_2$. 
In Sec. \ref{subsec:prob} we clarify the problem that we study in this paper.
In Sec. \ref{subsec:z2qht2qubit} we present our main results on this problem. Specifically, we derive the asymptotic behavior of the hypothesis testing relative entropy with parity restriction.
In Sec. \ref{subsec:critical} we derive the critical number of copies to achieve the goal of the asymmetric QHT.

\subsection{Problem description}
\label{subsec:prob}

Here, we specifically describe the problem that we focus on. Our research objects are two pure single qubit states. The purpose of the asymmetric QHT is to discriminate between such two qubit states with the minimal probability of type-II error while allowing the probability of type-I error less than a constant $\epsilon$. How fast the minimal probability declines in the number of copies is our main concern. However, different from the general QHT, the POVMs that can be performed in tests are only the symmetric POVMs corresponding to the resource theory of parity. Mathematically, the resource theory of parity is associated with the group $\mathbb{Z}_2$. Therefore, our problem can be summarized as figuring out the asymptotic behavior of the minimal probability of type-II error by performing the $\mathbb{Z}_2$-invariant POVMs, i.e. the asymptotic behavior of the minimal $\b_n$ as the number of copies $n$ increases in Fig. \ref{fig:asyqht}.

First, given the Hilbert space $\cH\cong\bbC^2$, denote by $\ket{0}$ and $\ket{1}$ a pair of standard states of even and odd parity respectively. Specifically, the parity transformation flips the sign of $\ket{1}$ only. Then two arbitrary pure qubit states generally read as
\beq
\label{eq:twopuregeneric}
\bal
\ket{\psi_0}&=\sqrt{p}\ket{0}+\sqrt{1-p}\ket{1},\\
\ket{\psi_1}&=\sqrt{q}\ket{0}+e^{i\phi}\sqrt{1-q}\ket{1},
\eal
\eeq
with $p,q\in[0,1]$ and $\phi\in[0,2\pi)$. Note here that we eliminate the initial phase of $\ket{\psi_0}$ by adjusting the phase of $\ket{1}$.  In what follows we shall denote $\rho_0=\proj{\psi_0}$ and $\rho_1=\proj{\psi_1}$. In the special case when one of $p,q$ is $0$ or $1$, the phase $\phi$ is redundant as we can absorb the phase into the basis element $\ket{0}$ or $\ket{1}$.

Second, we analyze the group $\mathbb{Z}_2$ and the corresponding operation $\mathbb{Z}_2$-twirling in detail. The elements of group $\mathbb{Z}_2$ are $e$ (identity) and $f$ (flip), and their representation on Hilbert space is
\beq
\label{eq:Z2SSR}
T(e)=I, \quad T(f)=\om,
\eeq
where $I$ is the identity operator, and $\om$ is the parity operator satisfying $\om^2=I$. Hence, $\om$ is a Hermitian operator and its eigenvalues are $\pm 1$ \cite{rfgour2008}.
Then the Hilbert space can be decomposed into the eigenspaces of even and odd parity respectively corresponding to the eigenvalues of parity operator $\om$, i.e.
$$\cH=\cH_{even}\oplus\cH_{odd}.$$
Thus, the action of operator $\om$ is defined by
\beq
\label{eq:Z2SSR-1}
\om(\ket{\psi})=
\begin{cases}
\ket{\psi}, & \ket{\psi}\in\cH_{even},\\
-\ket{\psi}, & \ket{\psi}\in\cH_{odd}.
\end{cases}
\eeq
For completeness and convenience, we further explain the action of operator $\omega$ on an $N$-partite tensor product state in $\cH^{\ox N}$. We denote the tensor product state as $\bigotimes_{i=1}^N\ket{l_i,m_i}$, $\forall ~l_i\in\{0,1\}$, where $\ket{l_i,m_i}\in\cH_{even}$ if $l_i=0$ and $\ket{l_i,m_i}\in\cH_{odd}$ if $l_i=1$. Then the action of operator $\omega$ is defined by
\beq
\label{eq:Z2SSR-1x}
\om\left(\bigotimes_{i=1}^N\ket{l_i,m_i}\right)=
(-1)^b \bigotimes_{i=1}^N\ket{l_i,m_i},
\eeq
where $b=\left(\sum\limits_{i=1}^N l_i\right)\mod 2$.

Based on the discussion in Sec. \ref{subsec:qrta}, the symmetric states here are those $\mathbb{Z}_2$-invariant states satisfying
\beq
\label{eq:z2invs}
\om\rho\om =\rho,\quad \text{or equivalently}\quad [\rho,\om]=0.
\eeq
Using the $\mathbb{Z}_2$-twirling operation we can equivalently express any symmetric state as 
\beq
\label{eq:z2twirlingdef}
\cZ[\rho]=\frac{1}{2}\rho+\frac{1}{2}\om\rho\om.
\eeq
for some state $\rho$.
The symmetric operations are those $\mathbb{Z}_2$-invariant operations satisfying
\beq
\label{eq:z2invs-op}
\om\circ\cE\circ\om =\cE.
\eeq
Thus, all symmetric binary POVMs can similarly be characterized by the $\mathbb{Z}_2$-twirling operation as
\beq
\label{eq:sympovm}
\{\cZ[E],\cZ[\I-E]\},\quad \forall~ 0\leq E\leq \I.
\eeq

Therefore, the two error probabilities arising from the symmetric POVMs are accordingly specified as 
\beq
\label{eq:errorprob-sym}
\bal
\a_n&=1-\tr(\rho_0^{\ox n}\cZ[E_n]),\\
\b_n&=\tr(\rho_1^{\ox n}\cZ[E_n]).
\eal
\eeq 
Because the $G$-twirling operation is a self-adjoint operator, we further obtain that 
$$\tr(\rho_i^{\ox n}\cZ[E_n])=\tr(\cZ[\rho_i^{\ox n}]E_n),~\forall i=0,1.$$
Due to this relation we may equivalently consider the hypothesis testing relative entropy between two symmetric states as follows:
\beq
\label{eq:qhtre-sym}
\bal
& D_H^\e(\cZ[\rho_0^{\ox n}]||\cZ[\rho_1^{\ox n}])\equiv \\
& -\log\inf_{0\leq E_n\leq \I_n}\big\{\tr(\cZ[\rho_1^{\ox n}]E_n): \tr(\cZ[\rho_0^{\ox n}]E_n)\geq 1-\e\big\}.
\eal
\eeq

Finally, our problem becomes to analytically determine the asymptotic behavior of $D_H^\e(\cZ[\rho_0^{\ox n}]||\cZ[\rho_1^{\ox n}])$ given by Eq. \eqref{eq:qhtre-sym} for two arbitrary pure qubit states $\rho_0$ and $\rho_1$. Note that the quantum Stein's Lemma, i.e. Lemma \ref{le:qstein}, characterizes the asymptotic behavior of the hypothesis testing relative entropy between two states in the form of multiple copies, namely $\rho_0^{\ox n}$ and $\rho_1^{\ox n}$.
However, the $\mathbb{Z}_2$-twirling on an $n$-copy state is generally no longer in the form of multiple copies. Thus, we cannot derive the asymptotic behavior of $D_H^\e(\cZ[\rho_0^{\ox n}]||\cZ[\rho_1^{\ox n}])$ by quantum Stein's Lemma directly. This requires us to derive the analytic expression of $\cZ[\rho^{\ox n}]$ for any number of copies. In the following subsection we focus on the pure qubit states and derive the analytic expression of $\cZ[\proj{\psi}^{\ox n}]$ for any pure qubit state $\ket{\psi}$ and any copies. By virtue of the analytic expression we derive the asymptotic behavior of the minimal probability of type-II error with symmetric POVMs. We will discuss our results in detail by comparing with quantum Stein's Lemma below Theorem \ref{th:main}.

\subsection{Asymptotic behavior of the minimal probability of type-II error}
\label{subsec:z2qht2qubit}

After clarifying the specific problem, we shall deeply study it in this subsection. We present our main results by Theorem \ref{th:main}, where we completely determine the asymptotic behavior of type-II error probability for arbitrary two pure qubit states $\rho_0$ and $\rho_1$. Theorem \ref{th:main} reveals a very interesting fact that in most instances the two pure qubit states prepared in finite copies can be perfectly distinguished according to the measurement results with symmetric POVMs. Theorem \ref{th:main} is composed of three parts which are supported by Lemmas \ref{pp:generaltwoqubit} - \ref{pp:reduced-2} respectively.

We begin with formulating the expression of $\mathbb{Z}_2$-twirling on the $n$ copies of a pure qubit state. The corresponding Hilbert space is $(\bbC^{2})^{\ox n}$. First of all, it is necessary to clarify the two subspaces $\cH_{even}$ and $\cH_{odd}$ of $(\bbC^{2})^{\ox n}$. According to the action of parity operator $\om$ defined by Eq. \eqref{eq:Z2SSR-1}, we conclude that $\cH_{even}$ is spanned by the vectors that up to a permutation of subsystems are $\ket{0}^{\ox (n-j)}\ket{1}^{\ox j}$ where $j$ is even, and $\cH_{odd}$ is spanned by the vectors that up to a permutation of subsystems are $\ket{0}^{\ox (n-j)}\ket{1}^{\ox j}$ where $j$ is odd. It is not difficult to calculate that $\dim(\cH_{even})=\sum\limits_{j-even}{n\choose j}$ and $\dim(\cH_{odd})=\sum\limits_{j-odd}{n\choose j}$. Next, we can calculate $\cZ[\proj{\psi}^{\ox n}]$ for a qubit state $\ket{\psi}$ by Eq. \eqref{eq:z2twirlingdef}. Obviously, $\proj{\psi}^{\ox n}$ is invariant under $\mathbb{Z}_2$-twirling operation if $\ket{\psi}=\ket{0}$ or $\ket{1}$. It is remaining to consider the generic qubit state $\ket{\psi}=\sqrt{p}\ket{0}+e^{i\phi}\sqrt{1-p}\ket{1}$ when $p\in(0,1)$. A direct calculation yields that
\beq
\label{eq:z2ssrpure-3}
\bal
\ket{\psi}^{\ox n}&=\sum_{j-even} e^{i(j\phi)} (\sqrt{p})^{n-j}(\sqrt{1-p})^j \sqrt{{n \choose j}} \ket{v_j}\\
                  &+\sum_{j-odd} e^{i(j\phi)} (\sqrt{p})^{n-j}(\sqrt{1-p})^j \sqrt{{n\choose j}} \ket{v_j},      
\eal
\eeq
where 
\beq
\label{eq:z2ssrpure-2}
\bal
\forall j,~\ket{v_j}&=\sqrt{\frac{1}{{n\choose j}}}\sum\ket{0}^{\ox (n-j)}\ket{1}^{\ox j},
\eal
\eeq
and the sum is over all the ways of having $(n-j)$ systems in state $\ket{0}$ and $j$ systems in state $\ket{1}$. It is worthy to note the following two useful equalities
\beq
\label{eq:nchoosej}
\bal
\sum_{j-even} {n \choose j} a^{n-j}b^j&=\frac{1}{2}\big((a+b)^n+(a-b)^n\big),\\
\sum_{j-odd} {n \choose j} a^{n-j}b^j&=\frac{1}{2}\big((a+b)^n-(a-b)^n\big).
\eal
\eeq
They will be used to simplify our calculation in what follows.
Denote $\Dp\equiv p-(1-p)=2p-1$, and let
\beq
\label{eq:z2ssrpure-1}
\bal
\ket{0_p}&:=\sum_{j-even} e^{i(j\phi)} \sqrt{\frac{{n\choose j}p^{n-j}(1-p)^j}{\frac{1}{2}(1+(\Dp)^n)}}\ket{v_j},\\
\ket{1_p}&:=\sum_{j-odd} e^{i(j\phi)} \sqrt{\frac{{n\choose j}p^{n-j}(1-p)^j}{\frac{1}{2}(1-(\Dp)^n)}}\ket{v_j}.
\eal
\eeq
Then one can rewrite Eq. \eqref{eq:z2ssrpure-3} as
\beq
\label{eq:z2ssrpure-3.x}
\ket{\psi}^{\ox n}=\sqrt{\frac{1+(\Dp)^n}{2}}\ket{0_p}+\sqrt{\frac{1-(\Dp)^n}{2}}\ket{1_p},
\eeq
and thus $\proj{\psi}^{\ox n}$ is formulated as
\beq
\label{eq:z2ssrpure-4}
\bal
&~~~~\frac{1+(\Dp)^n}{2}\proj{0_p}+\frac{1-(\Dp)^n}{2}\proj{1_p}\\
&+\sqrt{\frac{(1+(\Dp)^n)(1-(\Dp)^n)}{4}}(\ketbra{0_p}{1_p}+\ketbra{1_p}{0_p}).
\eal
\eeq
It follows from Eq. \eqref{eq:z2twirlingdef} that
\beq
\label{eq:ncopyz2twirling-1}
\cZ[\proj{\psi}^{\ox n}]=\frac{1+(\Dp)^n}{2}\proj{0_p}+\frac{1-(\Dp)^n}{2}\proj{1_p}.
\eeq
According to Eq. \eqref{eq:ncopyz2twirling-1}, one can verify that for generic pure qubit state $\ket{\psi}$, $\cZ[\proj{\psi}^{\ox n}]$ cannot be expressed as the $n$ copies of some qubit state. Therefore, we have to introduce the dual relation given by Eq. \eqref{eq:dualsdp} to deal with our problem.
In what follows, we shall adopt the expression of $\cZ[\proj{\psi}^{\ox n}]$ in Eq. \eqref{eq:ncopyz2twirling-1} when the $\mathbb{Z}_2$-twirling of an $n$-copy generic pure qubit state is needed.

Now we are ready to show our main results on the asymptotic behavior of the minimal probability of type-II error for two pure qubit states $\rho_0$ and $\rho_1$. In order to understand the main results in Theorem \ref{th:main} conveniently, we also illustrate them in Table \ref{tab:pure} as follows.

\begin{table*}[htbp]
\caption{QHT between two pure qubit states using $\mathbb{Z}_2$-invariant POVMs}
\centering
\begin{tabular}{|c|c|c|}
\hline
$\{\ket{\psi_0},~\ket{\psi_1}\}$ & $\cZ[\proj{\psi_0}^{\ox n}]$ and $\cZ[\proj{\psi_1}^{\ox n}]$ & $\inf\{\tr(\cZ[\rho_1^{\ox n}]E):\tr(\cZ[\rho_0^{\ox n}]E)\geq 1-\e\}$ \\
\hline
$\ket{\psi_0}=\ket{\psi_1}$ & equal & $1-\epsilon$ for any $n$ \\
\hline
$\{\sqrt{p}\ket{0}+\sqrt{1-p}\ket{1},~\sqrt{p}\ket{0}-\sqrt{1-p}\ket{1}\}$ & equal & $1-\epsilon$ for any $n$ \\
\hline
other cases & distinct & reduce to zero when $n\geq n_\epsilon$ for a finite $n_\epsilon$ \\
\hline
\end{tabular}
\label{tab:pure}
\end{table*}

\begin{theorem}
\label{th:main}
Suppose $\rho_0=\proj{\psi_0}$ and $\rho_1=\proj{\psi_1}$ where $\ket{\psi_0}$ and $\ket{\psi_1}$ are two pure qubit states.

(i) If $\ket{\psi_0}=\ket{\psi_1}$, or $\{\ket{\psi_0},\ket{\psi_1}\}=\{\sqrt{p}\ket{0}+\sqrt{1-p}\ket{1},\sqrt{p}\ket{0}-\sqrt{1-p}\ket{1}\}$ with $p\in(0,1)$, then for any $n$,
\beq
\label{eq:mainth-1}
\bal
&~~~~2^{-D_H^{\e}(\cZ[\rho_0^{\ox n}]||\cZ[\rho_1^{\ox n}])} \\
&\equiv\inf\{\tr(\cZ[\rho_1^{\ox n}]E):\tr(\cZ[\rho_0^{\ox n}]E)\geq 1-\e\} \\
&=1-\e.
\eal
\eeq

(ii) If $\ket{\psi_0}$ and $\ket{\psi_1}$ are not given in (i), then for any given constant $\e\in(0,1)$, there exists a large enough $n_\e$ such that $\forall n\geq n_\e$,
\beq
\label{eq:mainth-2}
\inf\{\tr(\cZ[\rho_1^{\ox n}]E):\tr(\cZ[\rho_0^{\ox n}]E)\geq 1-\e\}=0.
\eeq
Particularly, if one of $\ket{\psi_0}$ and $\ket{\psi_1}$ is $\ket{0}$, and the other one is $\ket{1}$, then $\forall n\geq 1$, Eq. \eqref{eq:mainth-2} holds.
\qed
\end{theorem}

\begin{proof}
(i) In the first case, namely when $\ket{\psi_0}=\ket{\psi_1}$, one can obtain Eq. \eqref{eq:mainth-1} directly by definition. In the second case, we may assume $\ket{\psi_0}=\sqrt{p}\ket{0}+\sqrt{1-p}\ket{1}$ and $\ket{\psi_1}=\sqrt{p}\ket{0}-\sqrt{1-p}\ket{1}$ without loss of generality. Then it follows that
\beq
\label{eq:mainthmi-4}
\bal
\cZ[\rho_0^{\ox n}]=\frac{1+(\Dp)^n}{2}\proj{0_p}+\frac{1-(\Dp)^n}{2}\proj{1_p},\\
\cZ[\rho_1^{\ox n}]=\frac{1+(\Dp)^n}{2}\proj{\tilde{0}_p}+\frac{1-(\Dp)^n}{2}\proj{\tilde{1}_p},\\
\eal
\eeq
where 
\beq
\label{eq:mainthmi-4.1}
\bal
\ket{0_p}&=\sum_{j-even} \sqrt{\frac{{n\choose j}p^{n-j}(1-p)^j}{\frac{1}{2}(1+(\Dp)^n)}}\ket{v_j},\\
\ket{1_p}&=\sum_{j-odd}  \sqrt{\frac{{n\choose j}p^{n-j}(1-p)^j}{\frac{1}{2}(1-(\Dp)^n)}}\ket{v_j},\\
\ket{\tilde{0}_p}&=\ket{0_p},\quad \ket{\tilde{1}_p}=-\ket{1_p},
\eal
\eeq
and $\ket{v_j}$ is given by Eq. \eqref{eq:z2ssrpure-2}. It follows from Eq. \eqref{eq:mainthmi-4} that $\forall n\geq 1$, $\cZ[\rho_0^{\ox n}]=\cZ[\rho_1^{\ox n}]$. Thus, by definition we obtain Eq. \eqref{eq:mainth-1}.

(ii) This assertion follows from Lemmas \ref{pp:generaltwoqubit} - \ref{pp:reduced-2}. For the last statement, one can verify it directly.

This completes the proof.
\end{proof}

From Eq. \eqref{eq:mainth-1} we observe that the minimum probability of type-II error does not decrease in the number of copies, if $\ket{\psi_0}=\ket{\psi_1}$, or $\ket{\psi_0}=\sqrt{p}\ket{0}+\sqrt{1-p}\ket{1}$ and $\ket{\psi_1}=\sqrt{p}\ket{0}-\sqrt{1-p}\ket{1}$ with $p\in(0,1)$. When $\ket{\psi_0}=\ket{\psi_1}$, the result is trivial. However, the same result for the case when $\ket{\psi_1}=\sqrt{p}\ket{0}-\sqrt{1-p}\ket{1}$ with $p\in(0,1)$ is interesting. It implies that $\mathbb{Z}_2$-twirling operation could make two orthogonal states indistinguishable for any $n$ copies. In other words, there exist two orthogonal states which can only be perfectly distinguished using asymmetric POVMs.
For example, let $p=\frac{1}{2}$. Such two pure qubit states are orthogonal. Nevertheless, after applying $\mathbb{Z}_2$-twirling operation they become indistinguishable as the two $\mathbb{Z}_2$-twirlings are the same. Instead, from Theorem \ref{th:main} (ii) we conclude that there exists some symmetric POVM as an optimal POVM such that the probability of type-II error reduces to zero in finite copies. 

Furthermore, by comparing our results in Theorem \ref{th:main} with quantum Stein's Lemma, it shows our results are more informative than the asymptotic behavior characterized by quantum Stein's Lemma, under the setting of $\mathbb{Z}_2$-invariant POVMs. For simplicity, we shall discuss according to the different cases listed in Table \ref{tab:pure}. First, when $\ket{\psi_0}=\ket{\psi_1}$, then $D(\rho_0||\rho_1)=0$, and the quantum Stein's Lemma is consistent with the result that the minimal probability of type-II error $\b_n(\epsilon)$ will maintain at some level even for an infinite number of copies. A slight difference in the setting of $\mathbb{Z}_2$-invariant POVMs is $\b_n(\epsilon)$ will maintain $1-\epsilon$ for any number of copies. Second, when $\{\ket{\psi_0},\ket{\psi_1}\}=\{\sqrt{p}\ket{0}+\sqrt{1-p}\ket{1},\sqrt{p}\ket{0}-\sqrt{1-p}\ket{1}\}$, it follows from Table \ref{tab:pure} that the minimal probability $\b_n(\epsilon)$ will also maintain $1-\epsilon$ for any number of copies in the setting of $\mathbb{Z}_2$-invariant POVMs, while quantum Stein's Lemma predicts it will converge to $0$ in the general setting. The comparison of this case reveals a qualitative difference. As we have discussed in the previous paragraph, to perfectly distinguish such two states $\sqrt{p}\ket{0}+\sqrt{1-p}\ket{1}$ and $\sqrt{p}\ket{0}-\sqrt{1-p}\ket{1}$, the implemented POVM must be an asymmetric one which is not invariant under the action of $\mathbb{Z}_2$-twirling. Third, in all other cases when $\ket{\psi_0}\neq \ket{\psi_1}$, the corresponding density matrices $\rho_0$ and $\rho_1$ have the same set of eigenvalues $\{0,1\}$, and $\supp(\rho_0)$ is generally not a subset of $\supp(\rho_1)$. It follows that $D(\rho_0||\rho_1)$ diverges, and quantum Stein's Lemma is consistent with the result that the minimal probability $\b_n(\epsilon)$ reduces to zero for a finite number of copies, which is also derived in the setting of $\mathbb{Z}_2$-invariant POVMs. It means we can complete the task of QHT between such two pure qubit states by implementing a $\mathbb{Z}_2$-invariant POVM, with finite resource. Therefore, our results are more suitable for practical scenarios than quantum Stein's Lemma.

Finally, we list the three essential lemmas to support Theorem \ref{th:main} (ii). In Lemma \ref{pp:generaltwoqubit} we consider the non-degenerate case, i.e., both $p,q\in(0,1)$. In Lemmas \ref{pp:reduced} and \ref{pp:reduced-2} we consider the two degenerate cases, i.e., either $p$ or $q$ belongs to $\{0,1\}$.

\begin{lemma}
\label{pp:generaltwoqubit}
Suppose $\ket{\psi_0}$ and $\ket{\psi_1}$ are two distinct pure qubit states which are expressed in Eq. \eqref{eq:twopuregeneric} with $p,q\in(0,1)$. Except the case when $p=q$ and $\phi=\pi$ both hold, for any given $\e\in(0,1)$,
there exits a large enough $n_\e$ such that $\forall n\geq n_\e$,
\beq
\label{eq:nonreduced-1.2}
\inf\{\tr(\cZ[\rho_1^{\ox n}]E):\tr(\cZ[\rho_0^{\ox n}]E)\geq 1-\e\}=0.
\eeq
\end{lemma}

As both $p,q\in(0,1)$, both $\rho_0^{\ox n}$ and $\rho_1^{\ox n}$ are indeed supported on the four-dimensional subspace $\lin\{\ket{0_p},\ket{1_p},\ket{0_q},\ket{1_q}\}$. Then we derive an orthonormal basis of such a four-dimensional subspace. In terms of this orthonormal basis we derive the optimal POVM namely $\{E_n,\I-E_n\}$ such that 
\beq
\label{eq:oppovm}
\lim_{n\to\infty}\tr(\cZ[\rho_0^{\ox n}]E_n)=1,
\eeq
while
\beq
\label{eq:oppovm-1}
\tr(\cZ[\rho_1^{\ox n}]E_n)=0~~\text{for any $n\geq 1$}.
\eeq
It leads to the conclusion given by Eq. \eqref{eq:nonreduced-1.2}.
We put the detailed proof of Lemma \ref{pp:generaltwoqubit} in Appendix \ref{sec:proof-l2}. Next, we consider the two degenerate cases.

\begin{lemma}
\label{pp:reduced}
Suppose $\ket{\psi_0}=\ket{0}$ or $\ket{1}$, and
$
\ket{\psi_1}=\sqrt{q}\ket{0}+\sqrt{1-q}\ket{1}
$
with $q\in(0,1)$.
Then for any given $\e\in(0,1)$, there exits a large enough $n_\e$ such that $\forall n\geq n_\e$,
\beq
\label{eq:reduced-1.2}
\inf\{\tr(\cZ[\rho_1^{\ox n}]E):\tr(\cZ[\rho_0^{\ox n}]E)\geq 1-\e\}=0.
\eeq
\end{lemma}

\begin{lemma}
\label{pp:reduced-2}
Suppose 
$
\ket{\psi_0}=\sqrt{p}\ket{0}+\sqrt{1-p}\ket{1}
$
with $p\in(0,1)$, and
$\ket{\psi_1}=\ket{0}$ or $\ket{1}$. 
Then for any given $\e\in(0,1)$, there exits a large enough $n_\e$ such that $\forall n\geq n_\e$,
\beq
\label{eq:reduced-2-1.2}
\inf\{\tr(\cZ[\rho_1^{\ox n}]E):\tr(\cZ[\rho_0^{\ox n}]E)\geq 1-\e\}=0.
\eeq
\end{lemma}

In Lemma \ref{pp:reduced} and Lemma \ref{pp:reduced-2}, $\rho_0^{\ox n}$ and $\rho_1^{\ox n}$ are indeed supported on the three-dimensional subspace. Using the similar technique in the proof of Lemma \ref{pp:generaltwoqubit} we can prove the two degenerate cases analogously. We also put the detailed proofs of Lemma \ref{pp:reduced} and Lemma \ref{pp:reduced-2} in Appendices \ref{sec:proof-l3} and \ref{sec:proof-l4} respectively.

To sum up, all possible cases in Theorem \ref{th:main} (ii) are included in Lemmas \ref{pp:generaltwoqubit} - \ref{pp:reduced-2}, and have been discussed. Hence, we have shown the validity of Theorem \ref{th:main}. It is worth noting that we also specify the optimal POVM $\{E_n,\I-E_n\}$ to minimize the probability of type-II error in the corresponding proofs of Lemmas \ref{pp:generaltwoqubit} - \ref{pp:reduced-2}.

In the following subsection we will investigate the minimal number of copies $n_\e$ such that $\tr(\cZ[\rho_0^{\ox n}]E_n)\geq 1-\e$ and $\tr(\cZ[\rho_1^{\ox n}]E_n)=0$ when the optimal POVM is performed. We call such minimal number of copies as the critical number of copies.

\subsection{The critical number of copies for pefect distinguishing}
\label{subsec:critical}

It follows from Theorem \ref{th:main} that in most cases we can perfectly determine that the system is prepared in the state $\rho_0$ from two distinct pure qubit states with a finite number of copies using a symmetric POVM. From the perspective of resource conservation, it is better to prepare as few copies of states as possible while we can perfectly distinguish between the two pure qubit states. It inspires us to further derive the critical number of copies in our asymmetric QHT task. Note that the critical number of copies depends on the optimal POVM selected. Here, we select the optimal POVM given in the proofs of Lemmas \ref{pp:generaltwoqubit} - \ref{pp:reduced-2} respectively. By fixing the optimal POVM, we minimize the number of copies $n_\e$ such that $\tr(\cZ[\rho_0^{\ox n}]E_n)\geq 1-\e$ and $\tr(\cZ[\rho_1^{\ox n}]E_n)=0$.

\begin{proposition}
\label{pp:ne}
Suppose $\ket{\psi_0}=\sqrt{p}\ket{0}+\sqrt{1-p}\ket{1}$ and $\ket{\psi_1}=\sqrt{q}\ket{0}+e^{i\phi}\sqrt{1-q}\ket{1}$ are two distinct pure qubit states except the case when $p=q$ and $\phi=\pi$. 

(i) Suppose $p\in(0,1)$, and $\ket{\psi_1}\in\{\ket{0},\ket{1}\}$. When $\ket{\psi_1}=\ket{0}$, the critical number of copies $n_\e$ is equal to $\lceil\log_{p}\e\rceil$. When $\ket{\psi_1}=\ket{1}$, the critical number of copies $n_\e$ is equal to $\lceil\log_{1-p}\e\rceil$.

(ii) In all other cases, the critical number of copies has the formula as 
\beq
\label{eq:ppne-1}
n_\e=\log_{\la_{max}^2}(\e)+\bo\left(\log_{\la_{max}^2}(\e)\right)~~ \text{as}~~ \e\to 0, 
\eeq
where $\la_{max}$ is calculated by
$$\sqrt{pq+(1-p)(1-q)+2\sqrt{pq(1-p)(1-q)}\abs{\cos\phi}}.$$ 
Particularly, if $\ket{\psi_0}=\ket{0}$ or $\ket{1}$ and $q=\frac{1}{2}$, the critical number of copies $n_\e$ is equal to $\lceil\log_{\frac{1}{2}}\e+1\rceil$.
\end{proposition}

We also present the detailed proof of Proposition \ref{pp:ne} in Appendix \ref{sec:proof-prop}. It is crucial to obtain the critical number of copies $n_\e$ as fewer copies require less resouce. Thus, Proposition \ref{pp:ne} indicates how much resource we need at least to achieve our goal in the hypothesis testing between two pure qubit states.


\section{Quantum Hypothesis testing between a pure and a maximally mixed qubit state with parity}
\label{sec:mixedhyre}

In this section we extend our study to the context of the QHT between a pure and a maximally mixed qubit state. Specifically, we replace one of the two pure qubit states in Sec. \ref{sec:z2ssr} with a maximally mixed qubit state, and consider how the hypothesis testing relative entropy asymptotically behaves between such two states. If the pure state is $\ket{0}$ or $\ket{1}$, then both the pure state and maximally mixed state are invariant under $\cZ_2$-twirling. In this case, one can obtain the hypothesis testing relative entropy directly from Stein's Lemma. Thus, we only need to consider the pure state in the superposition of $\ket{0}$ and $\ket{1}$. The asymptotic behavior of the minimal probability of type-II error is characterized by Theorem \ref{thm:pandm}. Similar to Theorem \ref{th:main}, we also describe the main results in Theorem \ref{thm:pandm} by Table \ref{tab:p+mm} for convenience.

\begin{table*}[htbp]
\normalsize
\caption{QHT between a pure and a maximally mixed qubit state using $\mathbb{Z}_2$-invariant POVMs}
\centering
\begin{tabular}{|c|c|c|}
\hline
$\rho_0$ & $\rho_1$ & $\inf\{\tr(\cZ[\rho_1^{\ox n}]E):\tr(\cZ[\rho_0^{\ox n}]E)\geq 1-\e\}$ \\
\hline
$\proj{\psi}$, $\ket{\psi}=\sqrt{p}\ket{0}+\sqrt{1-p}\ket{1}$ & $\frac{1}{2} I$ & $\lim\limits_{n\to\infty}\frac{D_H^{\e}(\cZ[\rho_0^{\ox n}]||\cZ[\rho_1^{\ox n}])}{n}=1$ \\
\hline
$\frac{1}{2} I$ & $\proj{\psi}$, $\ket{\psi}=\sqrt{p}\ket{0}+\sqrt{1-p}\ket{1}$ & reduce to zero when $n\geq \lceil\log(1/\e)\rceil+1$ \\
\hline
\end{tabular}
\label{tab:p+mm}
\end{table*}

\begin{theorem}
\label{thm:pandm}
Let $\ket{\psi}=\sqrt{p}\ket{0}+\sqrt{1-p}\ket{1}$ for $p\in(0,1)$.

(i) Suppose $\rho_0=\proj{\psi}$ and $\rho_1=\frac{1}{2}I$. Then 
\beq
\label{eq:pandm-main1}
\lim\limits_{n\to\infty}\frac{D_H^{\e}(\cZ[\rho_0^{\ox n}]||\cZ[\rho_1^{\ox n}])}{n}=1.
\eeq

(ii) Suppose $\rho_0=\frac{1}{2}I$ and $\rho_1=\proj{\psi}$. Then for any given $\e\in(0,1)$ and for all $n\geq \lceil\log(1/\e)\rceil+1$,
\beq
\label{eq:pandm-main2}
\bal
\inf\{\tr(\cZ[\rho_1^{\ox n}]E):\tr(\cZ[\rho_0^{\ox n}]E)\geq 1-\e\}=0.
\eal
\eeq
\end{theorem}

\begin{proof}
It follows from Eq. \eqref{eq:ncopyz2twirling-1} that $\cZ[\proj{\psi}^{\ox n}]=\frac{1+(\Dp)^n}{2}\proj{0_p}+\frac{1-(\Dp)^n}{2}\proj{1_p}$, where $\ket{0_p}$ and $\ket{1_p}$ are mutually orthonormal. Thus, one can expand $\{\ket{0_p},\ket{1_p}\}$ to an orthonormal basis of $(\bbC^{2})^{\ox n}$
\beq
\label{eq:pandmmbasis}
\{\ket{0_p},\ket{1_p},\ket{2_p},\cdots,\ket{(2^n-1)_p}\}.
\eeq
In terms of this orthonormal basis, we obtain that
\beq
\label{eq:pandmmbasis-1}
\cZ[(\frac{1}{2}I)^{\ox n}]=\frac{1}{2^n} I^{\ox n}=\frac{1}{2^n}\sum_{j=0}^{2^n-1}\proj{j_p}.
\eeq

(i) Suppose $\rho_0=\proj{\psi}$ and $\rho_1=\frac{1}{2}I$. According to the dual relation Eq. \eqref{eq:dualsdp}, we obtain that $\b_n(\e)\equiv\max_{r\geq 0}f_{n,\e}(r)$, where
\beq
\label{eq:pandmopt-1}
\bal
f_{n,\e}(r)&:=\frac{1}{2}\big[1+(1-2\e)r-\norm{r\cZ[\rho_0^{\ox n}]-\cZ[\rho_1^{\ox n}]}_1\big]\\
&=\frac{1}{2}\Big[1+(1-2\e)r-\Big(\abs{\frac{r(1+(\Dp)^n)}{2}-\frac{1}{2^n}}\\
&~~~~+\abs{\frac{r(1-(\Dp)^n)}{2}-\frac{1}{2^n}}+\frac{1}{2^n}(2^n-2)\Big)\Big]\\
&=\frac{1}{2}\Big[\frac{1}{2^{n-1}}+(1-2\e)r-\Big(\abs{\frac{r(1+(\Dp)^n)}{2}-\frac{1}{2^n}}\\
&~~~~+\abs{\frac{r(1-(\Dp)^n)}{2}-\frac{1}{2^n}}\Big)\Big].
\eal
\eeq

If $\Dp\geq 0$, we can further obtain that
\beq
\label{eq:pandmopt-2}
f_{n,\e}(r)=
\begin{cases}
(1-\e)r, & r\in [0,r_1],\\
\frac{1}{2^n}+\frac{(1-2\e-(\Dp)^n)r}{2}, & r\in [r_1,r_2],\\
\frac{1}{2^{n-1}}-\e r, & r\in[r_2,+\infty),
\end{cases}
\eeq
where $r_1=(\frac{1}{2^n})/(\frac{(1+(\Dp)^n)}{2})$ and $r_2=(\frac{1}{2^n})/(\frac{(1-(\Dp)^n)}{2})$.
A straightforward calculation yields that 
\beq
\label{eq:pandmopt-3}
\max_{r\geq 0} f_{n,\e}(r)=
\begin{cases}
\frac{1-(\Dp)^n-\e}{(1-(\Dp)^n)2^{n-1}}, & \e\in(0,1/2),\\
\frac{1-\e}{(1+(\Dp)^n)2^{n-1})}, & \e\in[1/2,1).
\end{cases}
\eeq
The first part, namely $\e\in(0,1/2)$, holds for sufficiently large $n$ such that $1-2\e-(\Dp)^n>0$. It follows that the hypothesis testing relative entropy between $\cZ[\rho_0^{\ox n}]$ and $\cZ[\rho_1^{\ox n}]$, i.e. $D_H^{\e}(\cZ[\rho_0^{\ox n}]||\cZ[\rho_1^{\ox n}])$, can be expressed as the following piecewise function:
\beq
\label{eq:pandmopt-4}
\begin{cases}
-\log\frac{1-(\Dp)^n-\e}{(1-(\Dp)^n)2^{n-1}}, & \e\in(0,1/2),\\
-\log\frac{1-\e}{(1+(\Dp)^n)2^{n-1})}, & \e\in[1/2,1).
\end{cases}
\eeq
Thus, by calculating the limit we obtain 
$$\lim_{n\to\infty}\frac{D_H^{\e}(\cZ[\rho_0^{\ox n}]||\cZ[\rho_1^{\ox n}])}{n}=1.$$

If $\Dp<0$, we can similarly obtain Eq. \eqref{eq:pandmopt-4} by just replacing $(\Dp)^n$ with $(\abs{\Dp})^n$. Thus, we also conclude that $\lim\limits_{n\to\infty}\frac{D_H^{\e}(\cZ[\rho_0^{\ox n}]||\cZ[\rho_1^{\ox n}])}{n}=1$ for $\Dp<0$. Therefore, assertion (i) is valid.

(ii) Suppose $\rho_0=\frac{1}{2}I$ and $\rho_1=\proj{\psi}$. One can choose the binary POVM as $\{E\equiv\sum_{j=2}^{2^n-1}\proj{j_p}, I^{\ox n}-E\}$. It follows that
\beq
\label{eq:pandmii-1}
\tr(\cZ[\rho_0^{\ox n}] E)=\frac{2^n-2}{2^n},~~\text{while}~~\tr(\cZ[\rho_1^{\ox n}] E)=0.
\eeq
A straightforward calculation yields that $\tr(\cZ[\rho_0^{\ox n}] E)\geq 1-\e$ for all $n\geq \lceil\log(1/\e)\rceil+1$. Thus, by definition assertion (ii) is valid. 

This completes the proof.
\end{proof}

Different from Theorem \ref{th:main}, the minimal probability of type-II error decreases in the number of copies when $\rho_1$ is a maximally mixed state, and the decreasing rate is given by Eq. \eqref{eq:pandm-main1}. However, when $\rho_0$ is a maximally mixed state, we derive a conclusion that there exists some symmetric POVM as an optimal POVM such that the probability of type-II error reduces to zero in a finite number of copies, which is similar to Theorem \ref{th:main} (ii). Furthermore, the critical number of copies in this case has been specified to $\lceil\log(1/\e)\rceil+1$ by fixing an optimal POVM.


\section{Conclusions}
\label{sec:con}

In this paper we investigated the asymmetric quantum hypothesis testing between two qubit states within the resource theory of parity. Differently from the general QHT, we considered the operational limitations arising from the lack of a reference frame for chirality. This reference frame is associated with the group $\bbZ_2$ consisting of the identity transformation and the parity transformation. According to quantum resource theory, the POVMs that can be adopted in our task of QHT are those that respect the symmetry of the problem, i.e. the $\bbZ_2$-invariant POVMs. Hence, the focused problem was to figure out how the minimal probability of type-II error, or equivalently the hypothesis testing relative entropy, asymptotically behaves as the number of copies increases, when the POVMs are required to be $\bbZ_2$-invariant. By virtue of an equivalent characterization of $\bbZ_2$-invariant operations called $\bbZ_2$-twirling, we transformed the problem to minimizing the probability given by $\tr(\cZ[\rho_1^{\ox n}]E)$ for an arbitrary POVM $\{E,\I-E\}$, where $\rho_1$ represents the alternative hypothesis. Then, via an optimization duality, this optimization problem is further equivalent to maximizing a simple function of one real variable. 

We started by considering the QHT between two pure qubit states. We first explicitly formulated the expression of the $\bbZ_2$-twirling on an arbitrary pure qubit state. By analyzing the specific optimization problem in terms of a one-variable function, we completely solved the focused problem for two pure qubit states by Theorem \ref{th:main}. It is worth noting that our results showed that the minimal probability of type-II error reduces to zero in a finite number of copies, if the $\bbZ_2$-twirlings of such two pure qubit states are different. Physically, it means that the observer can perfectly decide the prepared quantum system is in state $\rho_1$ by only finite copies. Furthermore, from the perspective of resource conservation, we also derived the critical number of copies such that the probability of type-II error reduces to zero in Proposition \ref{pp:ne}. Finally, we replaced one of the two pure states with a maximally mixed state as an attempt to generalize the results on two pure states to the context of mixed states. In this case, we characterized the asymptotic behavior of the minimal probability of type-II error by Theorem \ref{thm:pandm}. 

Here, we provide two possible directions for future research. First, it is interesting to further investigate the asymmetric QHT between two mixed qubit states, or more generally between two mixed qudit states, within the resource theory of parity. Second, we may consider the operational limitations arising from the lack of some quantum reference frame associated with a generic group $G$. That is, the performed POVMs in the task of QHT can only be $G$-invariant.

\section*{Acknowledgments}
This work was developed during YS's visit to Prof. Gilad Gour under the program of CSC. YS appreciates Prof. Gilad Gour very much for the valuable discussion and comments.
The authors are very grateful to Milán Mosonyi and Masahito Hayashi for their nice comments.
The authors also want to thank the anonymous reviewers for helping improve the quality of this work.
YS and LC were supported by the NNSF of China (Grant No. 11871089), and the Fundamental Research Funds for the Central Universities (Grant Nos. JUSRP123029, ZG216S2005). 
CMS acknowledges the support of the Natural Sciences and Engineering Research Council of Canada (NSERC) through the Discovery Grant “The power of quantum resources” RGPIN-2022-03025 and the Discovery Launch Supplement DGECR-2022-00119.


\appendix

\section{Proof of Lemma \ref{pp:generaltwoqubit}}
\label{sec:proof-l2}

\textbf{Proof of Lemma \ref{pp:generaltwoqubit}.}
For both $p,q\in(0,1)$, we obtain
\beq
\label{eq:generic-4}
\bal
&\cZ[\proj{\psi_0}^{\ox n}]=\frac{1+(\Dp)^n}{2}\proj{0_p}+\frac{1-(\Dp)^n}{2}\proj{1_p}\\
&\cZ[\proj{\psi_1}^{\ox n}]=\frac{1+(\Dq)^n}{2}\proj{0_q}+\frac{1-(\Dq)^n}{2}\proj{1_q},\\
\eal
\eeq
where 
\beq
\label{eq:generic-4.1}
\bal
\ket{0_p}&:=\sum_{j-even} \sqrt{\frac{{n\choose j}p^{n-j}(1-p)^j}{\frac{1}{2}(1+(\Dp)^n)}}\ket{v_j},\\
\ket{1_p}&:=\sum_{j-odd}  \sqrt{\frac{{n\choose j}p^{n-j}(1-p)^j}{\frac{1}{2}(1-(\Dp)^n)}}\ket{v_j},\\
\ket{0_q}&:=\sum_{j-even} e^{i(j\phi)} \sqrt{\frac{{n\choose j}q^{n-j}(1-q)^j}{\frac{1}{2}(1+(\Dq)^n)}}\ket{v_j},\\
\ket{1_q}&:=\sum_{j-odd} e^{i(j\phi)} \sqrt{\frac{{n\choose j}q^{n-j}(1-q)^j}{\frac{1}{2}(1-(\Dq)^n)}}\ket{v_j},
\eal
\eeq
and $\ket{v_j}$ is given by Eq. \eqref{eq:z2ssrpure-2}.
First of all, we need to find an orthonormal basis, namely $\{\ket{0_L},\ket{1_L},\ket{2_L},\ket{3_L}\}$, of $\lin\{\ket{0_p},\ket{1_p},\ket{0_q},\ket{1_q}\}$.
According to the Gram-Schmidt process we formulate the orthogonal vectors as
\beq
\label{eq:generic-5}
\bal
\ket{0_L}&:=\ket{0_q},\\
\ket{1_L}&:=\ket{1_q},\\
\ket{u_2}&=\ket{0_p}-\braket{0_q}{0_p}\ket{0_q},\\
\ket{u_3}&=\ket{1_p}-\braket{1_q}{1_p}\ket{1_q}.\\
\eal
\eeq
Denote $\mu_1=\sqrt{pq}$ and $\mu_2=\sqrt{(1-p)(1-q)}$. A calculation yields that 
\begin{widetext}
\beq
\label{eq:generic-5.1}
\bal
\braket{0_q}{0_p}&=\sqrt{\frac{1}{\frac{1}{4}(1+(\Dp)^n)(1+(\Dq)^n)}}\cdot\\
&\big(\sum_{j,j'-even}e^{-i(j\phi)}\sqrt{{n\choose j}{n\choose j'}q^{n-j}p^{n-j'}(1-q)^{j}(1-p)^{j'}}\braket{v_{j}}{v_{j'}}\big)\\
&=\sqrt{\frac{1}{\frac{1}{4}(1+(\Dp)^n)(1+(\Dq)^n)}}\sum_{j-even}{n\choose j}e^{-i(j\phi)}\mu_1^{n-j}\mu_2^j\\
&=\sqrt{\frac{1}{\frac{1}{4}(1+(\Dp)^n)(1+(\Dq)^n)}}\frac{(\mu_1+e^{-i\phi}\mu_2)^n+(\mu_1-e^{-i\phi}\mu_2)^n}{2}\\
&=\frac{(\mu_1+e^{-i\phi}\mu_2)^n+(\mu_1-e^{-i\phi}\mu_2)^n}{\sqrt{(1+(\Dp)^n)(1+(\Dq)^n)}}.
\eal
\eeq
\end{widetext}
Similarly, we obtain
\beq
\label{eq:generic-5.11}
\braket{1_q}{1_p}=\frac{(\mu_1+e^{-i\phi}\mu_2)^n-(\mu_1-e^{-i\phi}\mu_2)^n}{\sqrt{(1-(\Dp)^n)(1-\Dq)^n)}}.
\eeq
Let 
\beq
\label{eq:gen-5x}
\braket{0_q}{0_p}:=a_n e^{i\a_n}\quad \text{and}\quad \braket{1_q}{1_p}:=b_n e^{i\b_n},
\eeq 
where $a_n,b_n$ are the modulus of $\braket{0_q}{0_p}$ and $\braket{1_q}{1_p}$ respectively. Since
$$\norm{u_2}=\sqrt{1-a_n^2},\quad \norm{u_3}=\sqrt{1-b_n^2},$$
we conclude that $0\leq a_n,b_n\leq 1$. Next, we consider such two cases:
(i) $p=q$ and $\phi\neq 0,\pi$;
(ii) $p\neq q$.

Case (i) If $p=q$ and $\phi\neq 0,\pi$, both $\abs{p+e^{-i\phi}(1-p)}$ and $\abs{p-e^{-i\phi}(1-p)}$ are less than one. Thus, we obtain
$$\lim_{n\to\infty} a_n=\lim_{n\to\infty} b_n=0.$$

Case (ii) If $p\neq q$, we conclude that
\beq
\label{eq:generic-x5}
\bal
&~~~~\abs{\mu_1+e^{-i\phi}\mu_2}\leq\abs{\mu_1}+\abs{\mu_2}=\mu_1+\mu_2\\
&<\frac{p+q}{2}+\frac{1-p+1-q}{2}=1.
\eal
\eeq
Similarly one can verify that $\abs{\mu_1-e^{-i\phi}\mu_2}<1$. Thus, we obtain 
$$\lim_{n\to\infty} a_n=\lim_{n\to\infty} b_n=0.$$

Therefore, in both Cases (i) and (ii), $\lim\limits_{n\to\infty} a_n=\lim\limits_{n\to\infty} b_n=0$. Let
\beq
\label{eq:generic-5.3}
\bal
\ket{2_L}&=\frac{1}{\norm{u_2}}\ket{u_2}=\frac{1}{\sqrt{1-a_n^2}}\ket{0_p}-\frac{a_ne^{i\a_n}}{\sqrt{1-a_n^2}}\ket{0_q},\\
\ket{3_L}&=\frac{1}{\norm{u_3}}\ket{u_3}=\frac{1}{\sqrt{1-b_n^2}}\ket{1_p}-\frac{b_ne^{i\b_n}}{\sqrt{1-b_n^2}}\ket{1_q}.
\eal
\eeq 
It follows that
\beq
\label{eq:generic-5.4}
\bal
\ket{0_p}&=\sqrt{1-a_n^2}\ket{2_L}+a_ne^{i\a_n}\ket{0_L},\\
\ket{1_p}&=\sqrt{1-b_n^2}\ket{3_L}+b_ne^{i\b_n}\ket{1_L}.
\eal
\eeq 
Let $\sigma_0:=\cZ[\proj{\psi_0}^{\ox n}]$ and $\sigma_1:=\cZ[\proj{\psi_1}^{\ox n}]$. Write $\sigma_0$ and $\sigma_1$ in the matrix form as
\beq
\label{eq:generic-6.1}
\bal
\sigma_1&=
\bma
\frac{1+(\Dq)^n}{2} & 0 & 0 & 0 \\
0 & \frac{1-(\Dq)^n}{2} & 0 & 0 \\
0 & 0 & 0 & 0 \\
0 & 0 & 0 & 0 
\ema,\\
\sigma_0&=
\bma
m_{11} & 0 & m_{13} & 0 \\
0 & m_{22} & 0 & m_{24} \\
m_{13}^* & 0 & m_{33} & 0 \\
0 & m_{24}^* & 0 & m_{44}
\ema,
\eal
\eeq
where
\beq
\label{eq:generic-6.2}
\bal
&m_{11}=\frac{1+(\Dp)^n}{2}a_n^2, \quad m_{33}=\frac{1+(\Dp)^n}{2}(1-a_n^2),\\
&m_{13}=\frac{1+(\Dp)^n}{2}(a_n\sqrt{1-a_n^2}e^{i\a_n}),\\
&m_{22}=\frac{1-(\Dp)^n}{2}b_n^2,\quad m_{44}=\frac{1-(\Dp)^n}{2}(1-b_n^2),\\
&m_{24}=\frac{1-(\Dp)^n}{2}(b_n\sqrt{1-b_n^2}e^{i\b_n}). \\
\eal
\eeq
In order to obtain the minimum of the error probability defined in Eq. \eqref{eq:qhtre-sym}, we may take the POVM element to be $E=\proj{2_L}+\proj{3_L}$. We obtain that
\beq
\label{eq:generic-7.1}
\bal
\tr(\sigma_1E)&=0,\\
\tr(\sigma_0E)&=m_{33}+m_{44}.
\eal
\eeq
Since $\lim\limits_{n\to\infty}(\Dp)^n=0$, and $\lim\limits_{n\to\infty}a_n=\lim\limits_{n\to\infty}b_n=0$, we conclude that
\beq
\label{eq:generic-7.2}
\bal
&~~~~\lim_{n\to\infty} m_{33}+m_{44}\\
&=\lim_{n\to\infty}\Big(\frac{1+(\Dp)^n}{2}(1-a_n^2)+\frac{1-(\Dp)^n}{2}(1-b_n^2)\Big)\\
&=1.
\eal
\eeq
It implies that for any given $0<\e<1$, there exists a large enough $n_\e$ such that $\forall n\geq n_\e$, $\tr(\sigma_0 E)=m_{33}+m_{44}\geq 1-\e$. Thus, we obtain that $\forall n\geq n_\e$,
$$ 
\inf\{\tr(\cZ[\rho_1^{\ox n}]E):\tr(\cZ[\rho_0^{\ox n}]E)\geq 1-\e\}=0.
$$
The optimal POVM to reach the above infimum can be chosen as 
$$\{\proj{2_L}+\proj{3_L},\proj{0_L}+\proj{1_L}\}.$$

This completes the proof.
\qed

\section{Proof of Lemma \ref{pp:reduced}}
\label{sec:proof-l3}

\textbf{Proof of Lemma \ref{pp:reduced}.}
Since $q\in(0,1)$, it follows that 
\beq
\label{eq:reduced-3}
\cZ[\rho_1^{\ox n}]=\frac{1+(\Dq)^n}{2}\proj{0_q}+\frac{1-(\Dq)^n}{2}\proj{1_q},
\eeq
where $\ket{0_q}$ and $\ket{1_q}$ are expressed in Eq. \eqref{eq:z2ssrpure-1}.

Case (i) In this case we suppose $\ket{\psi_0}=\ket{0}$. Let $\ket{v_0}=\ket{0}^{\ox n}$. Then we use the Gram-Schmidt process to find an orthonormal basis, namely $\{\ket{0_L},\ket{1_L},\ket{2_L}\}$ for $\lin\{\ket{0_q},\ket{1_q},\ket{v_0}\}$. A straightforward calculation yields that 
\beq
\label{eq:reduced-4}
\bal
\ket{0_L}&:=\ket{0_q},\\
\ket{1_L}&:=\ket{1_q},\\
\ket{u_2}&:=\ket{v_0}-\braket{0_q}{v_0}\ket{0_q}\\
&=\ket{v_0}-\sqrt{\frac{q^n}{\frac{1}{2}(1+(\Dq)^n)}}\ket{0_q},\\
\ket{2_L}&:=\frac{1}{\norm{u_2}}\ket{u_2}\\
&=\sqrt{\frac{\frac{1}{2}(1+(\Dq)^n)}{\frac{1}{2}(1+(\Dq)^n)-q^n}}\ket{u_2}.
\eal
\eeq
It follows that 
$$
\ket{v_0}=\sqrt{\frac{\frac{1}{2}(1+(\Dq)^n)-q^n}{\frac{1}{2}(1+(\Dq)^n)}}\ket{2_L}+\sqrt{\frac{q^n}{\frac{1}{2}(1+(\Dq)^n)}}\ket{0_L}.
$$
Let $\sigma_0:=\cZ[\rho_0^{\ox n}]=\proj{v_0}$ and $\sigma_1:=\cZ[\rho_1^{\ox n}]$.
Denote 
\beq
\label{eq:xn-1.1}
x_n=\frac{q^n}{\frac{1}{2}(1+(\Dq)^n)},
\eeq and write $\sigma_0$ and $\sigma_1$ respectively in the matrix form as 
\beq
\label{eq:reducedcheck-1}
\bal
\sigma_0&=
\bma
x_n & 0 & \sqrt{(1-x_n)x_n}\\
0 & 0 & 0 \\
\sqrt{(1-x_n)x_n} & 0 & 1-x_n
\ema,\\
\sigma_1&=
\bma
\frac{1}{2}(1+(\Dq)^n) & 0 & 0\\
0 & \frac{1}{2}(1-(\Dq)^n) & 0 \\
0 & 0 & 0
\ema.
\eal
\eeq
Since $q\in(0,1)$, we obtain $\lim\limits_{n\to\infty} x_n=0$. Take the projector element to be $E=\proj{2_L}$. It follows that 
\beq
\label{eq:reducedcheck-2}
\tr(\sigma_0 E)=1-x_n(\to 1),\quad
\tr(\sigma_1 E)=0.
\eeq

Case (ii) In this case we suppose $\ket{\psi_0}=\ket{1}$. Let $\ket{v_n}=\ket{1}^{\ox n}$. Similarly, we can find an orthonormal basis of $\lin\{\ket{0_q},\ket{1_q},\ket{v_n}\}$ in the following process. For even $n$, a direct calculation yields that 
\beq
\label{eq:reduced-5}
\bal
\ket{0_L}&:=\ket{0_q},\\
\ket{1_L}&:=\ket{1_q},\\
\ket{u_2}&:=\ket{v_n}-\braket{0_q}{v_n}\ket{0_q}\\
&=\ket{v_n}-\sqrt{\frac{(1-q)^n}{\frac{1}{2}(1+(\Dq)^n)}}\ket{0_q},\\
\ket{2_L}&:=\frac{1}{\norm{u_2}}\ket{u_2}\\
&=\sqrt{\frac{\frac{1}{2}(1+(\Dq)^n)}{\frac{1}{2}(1+(\Dq)^n)-(1-q)^n}}\ket{u_2}.
\eal
\eeq
For odd $n$, a direct calculation yields that 
\beq
\label{eq:reduced-6}
\bal
\ket{0_L}&:=\ket{0_q},\\
\ket{1_L}&:=\ket{1_q},\\
\ket{u_2}&:=\ket{v_n}-\braket{1_q}{v_n}\ket{1_q}\\
&=\ket{v_n}-\sqrt{\frac{(1-q)^n}{\frac{1}{2}(1-(\Dq)^n)}}\ket{1_q},\\
\ket{2_L}&:=\frac{1}{\norm{u_2}}\ket{u_2}\\
&=\sqrt{\frac{\frac{1}{2}(1-(\Dq)^n)}{\frac{1}{2}(1-(\Dq)^n)-(1-q)^n}}\ket{u_2}.
\eal
\eeq
Then we denote 
\beq
\label{eq:xn-1.2}
x_n=\frac{(1-q)^n}{\frac{1}{2}(1+(-1)^n(\Dq)^n)}.
\eeq
Since $q\in(0,1)$, we obtain $\lim\limits_{n\to \infty}x_n=0$.
Then from \eqref{eq:reduced-5} - \eqref{eq:reduced-6} it follows that
\beq
\label{eq:reduced-7.1}
\ket{v_n}=
\begin{cases}
\sqrt{x_n}\ket{0_L}+\sqrt{1-x_n}\ket{2_L}, & \text{even $n$}, \\
\sqrt{x_n}\ket{1_L}+\sqrt{1-x_n}\ket{2_L}, & \text{odd $n$}. \\
\end{cases}
\eeq
Let $\sigma_0:=\cZ[\rho_0^{\ox n}]=\proj{v_n}$ and $\sigma_1:=\cZ[\rho_1^{\ox n}]$. For even $n$, we formulate $\sigma_0$ as
\beq
\label{eq:reduced-7.2-1}
\sigma_0=
\bma
x_n & 0 & e^{-i(n\phi)}\sqrt{(1-x_n)x_n}\\
0 & 0 & 0 \\
e^{i(n\phi)}\sqrt{(1-x_n)x_n} & 0 & 1-x_n
\ema, 
\eeq
and for odd $n$, we formulate $\sigma_0$ as
\beq
\label{eq:reduced-7.2-2}
\sigma_0=
\bma
0 & 0 & 0\\
0 & x_n & e^{-i(n\phi)}\sqrt{(1-x_n)x_n} \\
0 & e^{i(n\phi)}\sqrt{(1-x_n)x_n} & 1-x_n
\ema.
\eeq
For any $n\geq 1$, we can formulate $\sigma_1$ as
\beq
\label{eq:reduced-7.2-3}
\sigma_1=
\bma
\frac{1}{2}(1+(\Dq)^n) & 0 & 0\\
0 & \frac{1}{2}(1-(\Dq)^n) & 0 \\
0 & 0 & 0
\ema.
\eeq
In order to obtain the minimum of the error probability defined in Eq. \eqref{eq:qhtre-sym}, we may take the measurement element to be $E=\proj{2_L}$, where the expression of $\ket{2_L}$ is given by \eqref{eq:reduced-5} or \eqref{eq:reduced-6} depending on even $n$ or odd $n$. It follows that
\beq
\label{eq:reducedcheck-8}
\tr(\sigma_0 E)=1-x_n(\to 1),\quad
\tr(\sigma_1 E)=0.
\eeq

Combining \eqref{eq:reducedcheck-2} and \eqref{eq:reducedcheck-8} in Case (i) and Case (ii) respectively, we conclude that no matter $\ket{\psi_0}=\ket{0}$ or $\ket{1}$, there always exists a large enough $n_\e$ such that $\forall n\geq n_\e$, 
$$\inf\{\tr(\sigma_1 E):\tr(\sigma_0 E)\geq 1-\e\}=0.$$
To reach the above infimum, the POVM can be taken as $\{\proj{2_L}, \proj{0_L}+\proj{1_L}\}$.

This completes the proof.
\qed

\section{Proof of Lemma \ref{pp:reduced-2}}
\label{sec:proof-l4}

\textbf{Proof of Lemma \ref{pp:reduced-2}.}
The technique here is similar to that used in Lemma \ref{pp:reduced}. Let $\ket{v_0}=\ket{0}^{\ox n}$ and $\ket{v_n}=\ket{1}^{\ox n}$. First of all we need to find the orthonormal basis of $\lin\{\ket{0_p},\ket{1_p},\ket{v_0}\}$ or $\lin\{\ket{0_p},\ket{1_p},\ket{v_n}\}$ depending on $\ket{\psi_1}=\ket{0}$ or $\ket{1}$.

Case (i) In this case we consider $\ket{\psi_1}=\ket{0}$. A direct calculation yields that
\beq
\label{eq:reducedpsi1-1}
\bal
\ket{0_L}&:=\ket{v_0},\\
\ket{u_1}&:=\ket{0_p}-\braket{v_0}{0_p}\ket{v_0},\\
\ket{2_L}&:=\ket{1_p}\\
\ket{1_L}&:=\frac{1}{\norm{u_1}}\ket{u_1}\\
&=\sqrt{\frac{\frac{1}{2}(1+(\Dp)^n)}{\frac{1}{2}(1+(\Dp)^n)-p^n}}\ket{u_1}.
\eal
\eeq
It follows that
$$
\ket{0_p}=\sqrt{\frac{\frac{1}{2}(1+(\Dp)^n)-p^n}{\frac{1}{2}(1+(\Dp)^n)}}\ket{1_L}+\sqrt{\frac{p^n}{\frac{1}{2}(1+(\Dp)^n)}}\ket{0_L}.
$$
Denote 
\beq
\label{eq:xn-2.1}
x_n=\frac{p^n}{\frac{1}{2}(1+(\Dp)^n)}.
\eeq Let $\sigma_0=\cZ[\rho_0^{\ox n}]$ and  $\sigma_1=\cZ[\rho_1^{\ox n}]$. Then we write $\sigma_0$ and $\sigma_1$ in the matrix form as 
\beq
\label{eq:reducedpsi1-2}
\bal
\sigma_0&=
\left(
\begin{smallmatrix}
\frac{1+(\Dp)^n}{2}\cdot x_n & \frac{1+(\Dp)^n}{2}\cdot\sqrt{x_n(1-x_n)} & 0 \\
\frac{1+(\Dp)^n}{2}\cdot\sqrt{x_n(1-x_n)} & \frac{1+(\Dp)^n}{2}\cdot(1-x_n) & 0 \\
0 & 0 & \frac{1-(\Dp)^n}{2}
\end{smallmatrix}
\right),\\
\sigma_1&=
\bma
1 & 0 & 0 \\
0 & 0 & 0 \\
0 & 0 & 0 \\
\ema.
\eal
\eeq
Take the projector element to be $E=\proj{1_L}+\proj{2_L}$. It follows that
\beq
\label{eq:reduced2error-1}
\bal
\tr(\sigma_0E)&=1-\frac{(1+(\Dp)^n)}{2}\cdot x_n=1-p^n~ (\to 1),\\
\tr(\sigma_1E)&=0.
\eal
\eeq

Case (ii) In this case we consider $\ket{\psi_1}=\ket{1}$. For $n$ is even, a direct calculation yields that
\beq
\label{eq:reducedpsi1-3}
\bal
\ket{0_L}&:=\ket{v_n},\\
\ket{u_1}&:=\ket{0_p}-\braket{v_n}{0_p}\ket{v_n},\\
&=\ket{0_p}-\sqrt{\frac{(1-p)^n}{\frac{1}{2}(1+(\Dp)^n)}}\ket{v_n},\\
\ket{2_L}&:=\ket{1_p}\\
\ket{1_L}&:=\frac{1}{\norm{u_1}}\ket{u_1}\\
&=\sqrt{\frac{\frac{1}{2}(1+(\Dp)^n)}{\frac{1}{2}(1+(\Dp)^n)-(1-p)^n}}\ket{u_1}.
\eal
\eeq
Thus, $\{\ket{0_L},\ket{1_L},\ket{2_L}\}$ in \eqref{eq:reducedpsi1-3} is an orthonormal basis of $\lin\{\ket{v_n},\ket{0_p},\ket{1_p}\}$ when $n$ is even. For $n$ is odd, a direct calculation yields that
\beq
\label{eq:reducedpsi1-4}
\bal
\ket{0_L}&:=\ket{v_n},\\
\ket{1_L}&:=\ket{0_p}\\
\ket{u_2}&:=\ket{1_p}-\braket{v_n}{1_p}\ket{v_n},\\
&=\ket{1_p}-e^{in\phi}\sqrt{\frac{(1-p)^n}{\frac{1}{2}(1-(\Dp)^n)}}\ket{v_n},\\
\ket{2_L}&:=\frac{1}{\norm{u_2}}\ket{u_2}\\
&=\sqrt{\frac{\frac{1}{2}(1-(\Dp)^n)}{\frac{1}{2}(1-(\Dp)^n)-(1-p)^n}}\ket{u_2}.
\eal
\eeq
Thus, $\{\ket{0_L},\ket{1_L},\ket{2_L}\}$ in \eqref{eq:reducedpsi1-4} is an orthonormal basis of $\lin\{\ket{v_n},\ket{0_p},\ket{1_p}\}$ when $n$ is odd. Then we denote 
\beq
\label{eq:xn-2.2}
x_n=\frac{(1-p)^n}{\frac{1}{2}(1+(-1)^n(\Dp)^n)}.
\eeq
Since $p\in(0,1)$, we obtain $\lim_{n\to\infty}x_n=0$. It follows from \eqref{eq:reducedpsi1-3}-\eqref{eq:reducedpsi1-4} that
\beq
\label{eq:reducedpsi1-6}
\begin{cases}
\ket{0_p}=\sqrt{x_n}\ket{0_L}+\sqrt{1-x_n}\ket{1_L}, & \text{even $n$}, \\
\ket{1_p}=\sqrt{x_n}\ket{0_L}+\sqrt{1-x_n}\ket{2_L}, & \text{odd $n$}. \\
\end{cases}
\eeq
Let $\sigma_0=\cZ[\rho_0^{\ox n}]$ and  $\sigma_1=\cZ[\rho_1^{\ox n}]$. For even $n$ we formulate $\sigma_0$ as 
\beq
\label{eq:reducedpsi1-7.1}
\sigma_0=
\left(
\begin{smallmatrix}
\frac{1+(\Dp)^n}{2}\cdot x_n & \frac{1+(\Dp)^n}{2}\cdot\sqrt{x_n(1-x_n)} & 0 \\
\frac{1+(\Dp)^n}{2}\cdot \sqrt{x_n(1-x_n)} & \frac{1+(\Dp)^n}{2}\cdot(1-x_n) & 0 \\
0 & 0 & \frac{1-(\Dp)^n}{2}
\end{smallmatrix}
\right),
\eeq
and for odd $n$ we formulate $\sigma_0$ as
\beq
\label{eq:reducedpsi1-7.2}
\sigma_0=
\left(
\begin{smallmatrix}
\frac{1-(\Dp)^n}{2}\cdot x_n & 0 & \frac{1-(\Dp)^n}{2}\cdot\sqrt{x_n(1-x_n)} \\
0 & \frac{1+(\Dp)^n}{2} & 0 \\
\frac{1-(\Dp)^n}{2}\cdot \sqrt{x_n(1-x_n)} & 0 & \frac{1-(\Dp)^n}{2}\cdot(1-x_n)
\end{smallmatrix}
\right).
\eeq
For any $n\geq 1$, we can formulate $\sigma_1$ as
\beq
\label{eq:reducedpsi1-7.3}
\sigma_1=
\bma
1 & 0 & 0 \\
0 & 0 & 0 \\
0 & 0 & 0 \\
\ema.
\eeq
In order to obtain the minimum of the error probability defined in Eq. \eqref{eq:qhtre-sym}, we may take the measurement element to be $\proj{1_L}+\proj{2_L}$, where the expressions of $\ket{1_L}$ and $\ket{2_L}$ are given by \eqref{eq:reducedpsi1-3} or \eqref{eq:reducedpsi1-4} depending on even $n$ or odd $n$.
It follows that
\beq
\label{eq:reduced2error-2}
\tr(\sigma_0E)=1-(1-p)^n (\to 1),\quad \tr(\sigma_1E)=0.
\eeq

Combining the above two cases we conclude that for any given $\e\in(0,1)$, there exits a large enough $n_\e$ such that $\forall n\geq n_\e$,
$$
\inf\{\tr(\cZ[\rho_1^{\ox n}]E):\tr(\cZ[\rho_0^{\ox n}]E)\geq 1-\e\}=0.
$$

This completes the proof.
\qed

\section{Proof of Proposition \ref{pp:ne}}
\label{sec:proof-prop}

\textbf{Proof of Proposition \ref{pp:ne}.}
(i) This case corresponds to Lemma \ref{pp:reduced-2}. According to Eqs. \eqref{eq:reduced2error-1} and \eqref{eq:reduced2error-2} we determine the critical $n_\e$ such that $\forall n\geq n_\e$, $p^n\leq \e$ when $\ket{\psi_1}=\ket{0}$, and $(1-p)^n\leq \e$ when $\ket{\psi_1}=\ket{1}$. A direct calculation yields the assertion (i).

Next we consider the assertion (ii) in three cases.

(ii.a) Here, we consider $\ket{\psi_0}=\ket{0}$ or $\ket{1}$, and $\ket{\psi_1}=\sqrt{q}\ket{0}+\sqrt{1-q}\ket{1}$ with $q\in(0,1)$. This case corresponds to Lemma \ref{pp:reduced}. According to Eqs. \eqref{eq:reducedcheck-2} and \eqref{eq:reducedcheck-8} we determine the critical $n_\e$ such that $x_n\leq \e$, $\forall n\geq n_\e$, where $x_n$ is given by Eq. \eqref{eq:xn-1.1} or Eq. \eqref{eq:xn-1.2} depending on $\ket{\psi_0}=\ket{0}$ or $\ket{1}$. Specifically, if $\ket{\psi_0}=\ket{0}$, we shall consider the following inequality:
\beq
\label{ineq:ne-3}
\frac{q^n}{\frac{1}{2}(1+(\Dq)^n)}\leq \e.
\eeq
If $\Dq> 0$, i.e., $\frac{1}{2}<q<1$ it follows that $\forall n>1,~1<1+(\Dq)^n<1+\Dq=2q$. Then we obtain
\beq
\label{eq:ne-3.1}
\frac{q^{n_\e}}{2q}<\frac{q^{n_\e}}{1+(\Dq)^{n_\e}}<q^{n_\e}.
\eeq
Since $n_\e$ is an integer, we conclude that $\log_q(\e)+1\leq n_\e\leq\log_q(\e)+\log_q(\frac{1}{2})+1$.
One can verify $\lim\limits_{\e\to 0}\frac{n_\e}{\log_q\e}=1$. It implies that $n_\e=\log_q\e+\bo(\log_q\e)$ as $\e\to 0$. 
If $\Dq= 0$, i.e., $q=\frac{1}{2}$, it follows from \eqref{ineq:ne-3} that $n_\e=\lceil\log_{\frac{1}{2}}\e+1\rceil$.
If $\Dq< 0$, i.e., $0<q<\frac{1}{2}$, it follows that $\forall n>1,~2q=1+\Dq< 1+(\Dq)^n<2$. Then we obtain
\beq
\label{eq:ne-3.2}
\frac{q^{n_\e}}{2}<\frac{q^{n_\e}}{1+(\Dq)^{n_\e}}<\frac{q^{n_\e}}{2q}.
\eeq
Since $n_\e$ is an integer, we conclude that $\log_q(\e)\leq n_\e\leq\log_q(\e)+2$. It also implies that $n_\e=\log_q\e+\bo(\log_q\e)$ as $\e\to 0$.

If $\ket{\psi_0}=\ket{1}$, we shall consider the following inequality:
$$
\frac{(1-q)^n}{\frac{1}{2}(1+(-1)^n(\Dq)^n)}\leq \e.
$$
Similar to the above discussion, we conclude that if $\Dq>0$, then $\log_{1-q}(\e)\leq n_\e\leq \log_{1-q}(\e)+2$; if $\Dq=0$, then $n_\e=\lceil\log_{\frac{1}{2}}\e+1\rceil$; if $\Dq<0$, then $\log_{1-q}(\e)+1\leq n_\e\leq\log_{1-q}(\e)+\log_{1-q}(\frac{1}{2})+1$. Thus, when $\Dq\neq 0$, we conclude $n_\e=\log_{1-q}(\e)+\bo(\log_{1-q}(\e))$.

(ii.b) Here, we consider $\ket{\psi_0}=\sqrt{p}\ket{0}+\sqrt{1-p}\ket{1}$ and $\ket{\psi_1}=\sqrt{p}\ket{0}+e^{i\phi}\sqrt{1-p}\ket{1}$ with $\phi\neq 0,\pi$. This corresponds to Lemma \ref{pp:generaltwoqubit}. According to Eq. \eqref{eq:generic-7.2} we should determine the critical $n_\e$ such that $\forall n\geq n_\e$,
\beq
\label{eq:gnericne}
\frac{1+(\Dp)^n}{2}\cdot a_n^2+\frac{1-(\Dp)^n}{2}\cdot b_n^2\leq \e,
\eeq
where $a_n$ and $b_n$ are given by Eq. \eqref{eq:gen-5x}. 
Note that $p=q$ and $\phi\neq 0,\pi$ in this case. It follows that 
\beq
\label{eq:ineq-1.1}
\bal
a_n&=\frac{\abs{(p+e^{-i\phi}(1-p))^n+(p-e^{-i\phi}(1-p))^n}}{1+(\Dp)^n},\\
b_n&=\frac{\abs{(p+e^{-i\phi}(1-p))^n-(p-e^{-i\phi}(1-p))^n}}{1-(\Dp)^n}.
\eal
\eeq
Denote the two modulus as 
\beq
\label{eq:ineq-1.2}
\bal
\la_1&:=\abs{p+e^{-i\phi}(1-p)}\\
&=\sqrt{p^2+(1-p)^2+2p(1-p)\cos\phi},\\
\la_2&:=\abs{p-e^{-i\phi}(1-p)}\\
&=\sqrt{p^2+(1-p)^2-2p(1-p)\cos\phi}.
\eal
\eeq

We first consider the case when $\cos\phi=0$. It implies that $\la_1=\la_2$. Let $\la\equiv\la_1=\la_2$. Then we write $p+e^{-i\phi}(1-p)=\la e^{i\t}$ and $p-e^{-i\phi}(1-p)=\la e^{-i\t}$. 
It follows from Eq. \eqref{eq:ineq-1.1} that
\beq
\label{eq:ineq-1.3-1}
\bal
a_n&=\frac{\abs{\la^ne^{in\t}+\la^ne^{-in\t}}}{1+(\Dp)^n}=\frac{2\la^n\cos(n\t)}{1+(\Dp)^n},\\
b_n&=\frac{\abs{\la^ne^{in\t}-\la^ne^{-in\t}}}{1-(\Dp)^n}=\frac{2\la^n\sin(n\t)}{1-(\Dp)^n}.
\eal
\eeq
On the one hand, we obtain that  
\beq
\label{eq:ineq-1.3-2}
\bal
&~~~~\frac{1+(\Dp)^n}{2}\cdot a_n^2+\frac{1-(\Dp)^n}{2}\cdot b_n^2\\
&=\frac{2\la^{2n}\cos^2(n\t)}{1+(\Dp)^n}+\frac{2\la^{2n}\sin^2(n\t)}{1-(\Dp)^n}\\
&\geq\la^{2n}.
\eal
\eeq
On the other hand, we obtain that
\beq
\label{eq:ineq-1.3-3}
\bal
&~~~~\frac{1+(\Dp)^n}{2}\cdot a_n^2+\frac{1-(\Dp)^n}{2}\cdot b_n^2\\
&=\frac{2\la^{2n}\cos^2(n\t)}{1+(\Dp)^n}+\frac{2\la^{2n}\sin^2(n\t)}{1-(\Dp)^n}\\
&\leq \frac{\la^{2n}}{p(1-p)}.
\eal
\eeq
Since $n_\e$ is an integer, it follows that $\log_{\la^2}(\e)\leq n_{\e}\leq \log_{\la^2}(p(1-p)\e)+1$. It implies $\lim\limits_{\e\to 0}\frac{n_\e}{\log_{\la^2}(\e)}=1$, and thus $n_\e=\log_{\la^2}(\e)+\bf{o}(\log_{\la^2}(\e))$ as $\e\to 0$.

Second we consider the case when $\cos\phi\neq 0$. It implies that $\la_1\neq\la_2$. We write $p+e^{-i\phi}(1-p)=\la_1e^{i\t_1}$ and $p-e^{-i\phi}(1-p)=\la_2e^{i\t_2}$.
It follows from Eq. \eqref{eq:ineq-1.1} that
\beq
\label{eq:ineq-1.3}
\bal
a_n&=\frac{\abs{\la_1^ne^{in\t_1}+\la_2^ne^{in\t_2}}}{1+(\Dp)^n},\\
b_n&=\frac{\abs{\la_1^ne^{in\t_1}-\la_2^ne^{in\t_2}}}{1-(\Dp)^n}.
\eal
\eeq
From the triangle inequality we obtain that 
\beq
\label{eq:ineq-1.4}
\bal
\frac{\abs{\la_1^n-\la_2^n}}{1+(\Dp)^n} \leq a_n \leq \frac{\la_1^n+\la_2^n}{1+(\Dp)^n},\\
\frac{\abs{\la_1^n-\la_2^n}}{1-(\Dp)^n} \leq b_n \leq \frac{\la_1^n+\la_2^n}{1-(\Dp)^n}.\\
\eal
\eeq
It follows that
\beq
\label{eq:ineq-1.5}
\bal
&~~~~\frac{(\la_1^n-\la_2^n)^2}{2(1+(\Dp)^n)}+\frac{(\la_1^n-\la_2^n)^2}{2(1-(\Dp)^n)}\\
&\leq\frac{1+(\Dp)^n}{2}\cdot a_n^2+\frac{1-(\Dp)^n}{2}\cdot b_n^2\\
&\leq\frac{(\la_1^n+\la_2^n)^2}{2(1+(\Dp)^n)}+\frac{(\la_1^n+\la_2^n)^2}{2(1-(\Dp)^n)}.
\eal
\eeq
Let $\la_{max}:=\max\{\la_1,\la_2\}$ and $\la_{min}:=\min\{\la_1,\la_2\}$. It follows that
$\la_{max}=\sqrt{p^2+(1-p)^2+2p(1-p)\abs{\cos\phi}}$
and
$\la_{min}=\sqrt{p^2+(1-p)^2-2p(1-p)\abs{\cos\phi}}$.
On the one hand, we obtain that
\beq
\label{eq:ineq-1.6}
\bal
&~~~~\frac{1+(\Dp)^n}{2}\cdot a_n^2+\frac{1-(\Dp)^n}{2}\cdot b_n^2\\
&\geq \frac{(\la_1^n-\la_2^n)^2}{2(1+(\Dp)^n)}+\frac{(\la_1^n-\la_2^n)^2}{2(1-(\Dp)^n)}\\
&=\frac{(\la_1^n-\la_2^n)^2}{(1+(\Dp)^n)(1-(\Dp)^n)}\\
&> (\la_1^n-\la_2^n)^2>(1-\frac{\la_{min}}{\la_{max}})^2\la_{max}^{2n}.
\eal
\eeq
On the other hand, we obtain that 
\beq
\label{eq:ineq-1.7}
\bal
&~~~~\frac{1+(\Dp)^n}{2}\cdot a_n^2+\frac{1-(\Dp)^n}{2}\cdot b_n^2\\
&\leq \frac{(\la_1^n+\la_2^n)^2}{2(1+(\Dp)^n)}+\frac{(\la_1^n+\la_2^n)^2}{2(1-(\Dp)^n)}\\
&=\frac{(\la_1^n+\la_2^n)^2}{(1+(\Dp)^n)(1-(\Dp)^n)}\\
&\leq \frac{(\la_1^n+\la_2^n)^2}{(1+\Dp)(1-\Dp)}< \frac{\la_{max}^{2n}}{p(1-p)}.
\eal
\eeq
Since $n_\e$ is an integer, it follows from \eqref{eq:ineq-1.6}-\eqref{eq:ineq-1.7} that $\log_{\la_{max}^2}\big((1-\frac{\la_{min}}{\la_{max}})^{-2}\e\big)\leq n_\e\leq \log_{\la_{max}^2}(p(1-p)\e)+1$.
It implies $\lim\limits_{\e\to 0}\frac{n_\e}{\log_{\la_{max}^2}(\e)}=1$, and thus 
$n_\e=\log_{\la_{max}^2}(\e)+\bo\big(\log_{\la_{max}^2}(\e)\big)$.

(ii.c) Here we consider $\ket{\psi_0}=\sqrt{p}\ket{0}+\sqrt{1-p}\ket{1}$ and $\ket{\psi_1}=\sqrt{q}\ket{0}+e^{i\phi}\sqrt{1-q}\ket{1}$ with distinct $p,q$ and both of them in $(0,1)$. This also corresponds to Lemma \ref{pp:generaltwoqubit}. We shall keep considering the inequality \eqref{eq:gnericne}. Note that $p\neq q$ in this case. It follows that
\beq
\label{eq:ineq-2.1}
\bal
a_n&=\frac{\abs{(\mu_1+e^{-i\phi}\mu_2)^n+(\mu_1-e^{-i\phi}\mu_2)^n}}{\sqrt{(1+(\Dp)^n)(1+(\Dq)^n)}},\\
b_n&=\frac{\abs{(\mu_1+e^{-i\phi}\mu_2)^n-(\mu_1-e^{-i\phi}\mu_2)^n}}{\sqrt{(1-(\Dp)^n)(1-(\Dq)^n)}},
\eal
\eeq
where $\mu_1=\sqrt{pq}$ and $\mu_2=\sqrt{(1-p)(1-q)}$.
Denote the two modulus as 
\beq
\label{eq:ineq-2.2}
\bal
\la_1:=\abs{\mu_1+e^{-i\phi}\mu_2}=\sqrt{\mu_1^2+\mu_2^2+2\mu_1\mu_2\cos\phi},\\
\la_2:=\abs{\mu_1-e^{-i\phi}\mu_2}=\sqrt{\mu_1^2+\mu_2^2-2\mu_1\mu_2\cos\phi}.
\eal
\eeq

We first consider the case when $\cos\phi=0$. It implies that $\la_1=\la_2$. Let $\la\equiv\la_1=\la_2$. We write $\mu_1+e^{-i\phi}\mu_2=\la e^{i\t}$ and $\mu_1-e^{-i\phi}\mu_2=\la e^{-i\t}$. From Eq. \eqref{eq:ineq-2.1} we obtain
\beq
\label{eq:ineq-2.3-1}
\bal
a_n&=\frac{\abs{\la^ne^{in\t}+\la^n e^{-in\t}}}{\sqrt{(1+(\Dp)^n)(1+(\Dq)^n)}}\\
&=\frac{2\la^n\cos(n\t)}{\sqrt{(1+(\Dp)^n)(1+(\Dq)^n)}},\\
b_n&=\frac{\abs{\la^ne^{in\t}-\la^n e^{-in\t}}}{\sqrt{(1-(\Dp)^n)(1-(\Dq)^n)}}\\
&=\frac{2\la^n\sin(n\t)}{\sqrt{(1-(\Dp)^n)(1-(\Dq)^n)}}.
\eal
\eeq
On the one hand, we obtain that
\beq
\label{eq:ineq-2.3-2}
\bal
&~~~~\frac{1+(\Dp)^n}{2}\cdot a_n^2+\frac{1-(\Dp)^n}{2}\cdot b_n^2\\
&=\frac{2\la^{2n}\cos^2(n\t)}{1+(\Dq)^n}+\frac{2\la^{2n}\sin^2(n\t)}{1-(\Dq)^n}\geq\la^{2n}.
\eal
\eeq
On the other hand, we obtain that
\beq
\label{eq:ineq-2.3-3}
\bal
&~~~~\frac{1+(\Dp)^n}{2}\cdot a_n^2+\frac{1-(\Dp)^n}{2}\cdot b_n^2\\
&=\frac{2\la^{2n}\cos^2(n\t)}{1+(\Dq)^n}+\frac{2\la^{2n}\sin^2(n\t)}{1-(\Dq)^n}\leq\frac{\la^{2n}}{q(1-q)}.
\eal
\eeq
Since $n_\e$ is an integer, it follows that $\log_{\la^2}(\e)\leq n_{\e}\leq \log_{\la^2}(q(1-q)\e)+1$. It implies $\lim\limits_{\e\to 0}\frac{n_\e}{\log_{\la^2}(\e)}=1$, and thus $n_\e=\log_{\la^2}(\e)+\bf{o}(\log_{\la^2}(\e))$ as $\e\to 0$.

Second we consider the case when $\cos\phi\neq 0$. It implies $\la_1\neq \la_2$. We write $\mu_1+e^{-i\phi}\mu_2=\la_1e^{i\t_1}$ and $\mu_1-e^{-i\phi}\mu_2=\la_2e^{i\t_2}$. It follows from Eq. \eqref{eq:ineq-2.1} that 
\beq
\label{eq:ineq-2.3}
\bal
a_n&=\frac{\abs{\la_1^ne^{in\t_1}+\la_2^ne^{in\t_2}}}{\sqrt{(1+(\Dp)^n)(1+(\Dq)^n)}},\\
b_n&=\frac{\abs{\la_1^ne^{in\t_1}-\la_2^ne^{in\t_2}}}{\sqrt{(1-(\Dp)^n)(1-(\Dq)^n)}}.
\eal
\eeq
From the triangle inequality we obtain that
\beq
\label{eq:ineq-2.4}
\bal
\abs{\la_1^n-\la_2^n}\leq \sqrt{(1+(\Dp)^n)(1+(\Dq)^n)}\cdot a_n \leq \la_1^n+\la_2^n,\\
\abs{\la_1^n-\la_2^n}\leq \sqrt{(1-(\Dp)^n)(1-(\Dq)^n)}\cdot b_n \leq \la_1^n+\la_2^n.
\eal
\eeq
Let $\la_{max}:=\max\{\la_1,\la_2\}$ and $\la_{min}:=\min\{\la_1,\la_2\}$. It follows that
$\la_{max}=\sqrt{\mu_1^2+\mu_2^2+2\mu_1\mu_2\abs{\cos\phi}}$
and
$\la_{min}=\sqrt{\mu_1^2+\mu_2^2-2\mu_1\mu_2\abs{\cos\phi}}$.
On the one hand, we obtain that 
\beq
\label{eq:ineq-2.5}
\bal
&~~~~\frac{1+(\Dp)^n}{2}\cdot a_n^2+\frac{1-(\Dp)^n}{2}\cdot b_n^2\\
&\geq \frac{(\la_1^n-\la_2^n)^2}{2(1+(\Dq)^n)}+\frac{(\la_1^n-\la_2^n)^2}{2(1-(\Dq)^n)}\\
&=\frac{(\la_1^n-\la_2^n)^2}{(1+(\Dq)^n)(1-(\Dq)^n)}\\
&> (\la_1^n-\la_2^n)^2>(1-\frac{\la_{min}}{\la_{max}})^2\la_{max}^{2n}.
\eal
\eeq
On the other hand, we obtain that
\beq
\label{eq:ineq-2.6}
\bal
&~~~~\frac{1+(\Dp)^n}{2}\cdot a_n^2+\frac{1-(\Dp)^n}{2}\cdot b_n^2\\
&\leq \frac{(\la_1^n+\la_2^n)^2}{2(1+(\Dq)^n)}+\frac{(\la_1^n+\la_2^n)^2}{2(1-(\Dq)^n)}\\
&=\frac{(\la_1^n+\la_2^n)^2}{(1+(\Dq)^n)(1-(\Dq)^n)}\\
&\leq \frac{(\la_1^n+\la_2^n)^2}{(1+\Dq)(1-\Dq)}< \frac{\la_{max}^{2n}}{q(1-q)}.
\eal
\eeq
Since $n_\e$ is an integer, we conclude that $\log_{\la_{max}^2}\big((1-\frac{\la_{min}}{\la_{max}})^{-2}\e\big)\leq n_\e\leq \log_{\la_{max}^2}(q(1-q)\e)+1$.
It follows that $\lim\limits_{\e\to 0}\frac{n_\e}{\log_{\la_{max}^2}(\e)}=1$, and thus $n_\e=\log_{\la_{max}^2}(\e)+\bo\big(\log_{\la_{max}^2}(\e)\big)$.

This completes the proof.
\qed


\bibliography{qht}

\begin{thebibliography}{30}%
\makeatletter
\providecommand \@ifxundefined [1]{%
 \@ifx{#1\undefined}
}%
\providecommand \@ifnum [1]{%
 \ifnum #1\expandafter \@firstoftwo
 \else \expandafter \@secondoftwo
 \fi
}%
\providecommand \@ifx [1]{%
 \ifx #1\expandafter \@firstoftwo
 \else \expandafter \@secondoftwo
 \fi
}%
\providecommand \natexlab [1]{#1}%
\providecommand \enquote  [1]{``#1''}%
\providecommand \bibnamefont  [1]{#1}%
\providecommand \bibfnamefont [1]{#1}%
\providecommand \citenamefont [1]{#1}%
\providecommand \href@noop [0]{\@secondoftwo}%
\providecommand \href [0]{\begingroup \@sanitize@url \@href}%
\providecommand \@href[1]{\@@startlink{#1}\@@href}%
\providecommand \@@href[1]{\endgroup#1\@@endlink}%
\providecommand \@sanitize@url [0]{\catcode `\\12\catcode `\$12\catcode `\&12\catcode `\#12\catcode `\^12\catcode `\_12\catcode `\%12\relax}%
\providecommand \@@startlink[1]{}%
\providecommand \@@endlink[0]{}%
\providecommand \url  [0]{\begingroup\@sanitize@url \@url }%
\providecommand \@url [1]{\endgroup\@href {#1}{\urlprefix }}%
\providecommand \urlprefix  [0]{URL }%
\providecommand \Eprint [0]{\href }%
\providecommand \doibase [0]{https://doi.org/}%
\providecommand \selectlanguage [0]{\@gobble}%
\providecommand \bibinfo  [0]{\@secondoftwo}%
\providecommand \bibfield  [0]{\@secondoftwo}%
\providecommand \translation [1]{[#1]}%
\providecommand \BibitemOpen [0]{}%
\providecommand \bibitemStop [0]{}%
\providecommand \bibitemNoStop [0]{.\EOS\space}%
\providecommand \EOS [0]{\spacefactor3000\relax}%
\providecommand \BibitemShut  [1]{\csname bibitem#1\endcsname}%
\let\auto@bib@innerbib\@empty
\bibitem [{\citenamefont {Kendall}\ \emph {et~al.}(1994)\citenamefont {Kendall}, \citenamefont {Stuart}, \citenamefont {Ord}, \citenamefont {Forster}, \citenamefont {Arnold},\ and\ \citenamefont {O'Hagan}}]{kendall1994}%
  \BibitemOpen
  \bibfield  {author} {\bibinfo {author} {\bibfnamefont {M.}~\bibnamefont {Kendall}}, \bibinfo {author} {\bibfnamefont {A.}~\bibnamefont {Stuart}}, \bibinfo {author} {\bibfnamefont {K.}~\bibnamefont {Ord}}, \bibinfo {author} {\bibfnamefont {J.}~\bibnamefont {Forster}}, \bibinfo {author} {\bibfnamefont {S.}~\bibnamefont {Arnold}},\ and\ \bibinfo {author} {\bibfnamefont {A.}~\bibnamefont {O'Hagan}},\ }\href {https://books.google.ca/books?id=3h1W1JCBvN0C} {\emph {\bibinfo {title} {Kendall's Advanced Theory of Statistics, Classical Inference and the Linear Model}}},\ A Hodder Arnold Publication\ (\bibinfo  {publisher} {Wiley},\ \bibinfo {year} {1994})\BibitemShut {NoStop}%
\bibitem [{\citenamefont {Lehmann}\ and\ \citenamefont {Romano}(2005)}]{2005testing}%
  \BibitemOpen
  \bibfield  {author} {\bibinfo {author} {\bibfnamefont {E.~L.}\ \bibnamefont {Lehmann}}\ and\ \bibinfo {author} {\bibfnamefont {J.~P.}\ \bibnamefont {Romano}},\ }\href@noop {} {\emph {\bibinfo {title} {Testing statistical hypotheses}}},\ Springer Texts in Statistics\ (\bibinfo  {publisher} {Springer},\ \bibinfo {address} {New York},\ \bibinfo {year} {2005})\ p.\ \bibinfo {pages} {784}\BibitemShut {NoStop}%
\bibitem [{\citenamefont {Hiai}\ and\ \citenamefont {Petz}(1991)}]{hp1991}%
  \BibitemOpen
  \bibfield  {author} {\bibinfo {author} {\bibfnamefont {F.}~\bibnamefont {Hiai}}\ and\ \bibinfo {author} {\bibfnamefont {D.}~\bibnamefont {Petz}},\ }\href {https://doi.org/10.1007/BF02100287} {\bibfield  {journal} {\bibinfo  {journal} {Commun. Math. Phys.}\ }\textbf {\bibinfo {volume} {143}},\ \bibinfo {pages} {99} (\bibinfo {year} {1991})}\BibitemShut {NoStop}%
\bibitem [{\citenamefont {Audenaert}\ \emph {et~al.}(2008)\citenamefont {Audenaert}, \citenamefont {Nussbaum}, \citenamefont {Szkola},\ and\ \citenamefont {Verstraete}}]{erate2008}%
  \BibitemOpen
  \bibfield  {author} {\bibinfo {author} {\bibfnamefont {K.~M.~R.}\ \bibnamefont {Audenaert}}, \bibinfo {author} {\bibfnamefont {M.}~\bibnamefont {Nussbaum}}, \bibinfo {author} {\bibfnamefont {A.}~\bibnamefont {Szkola}},\ and\ \bibinfo {author} {\bibfnamefont {F.}~\bibnamefont {Verstraete}},\ }\href {https://doi.org/10.1007/s00220-008-0417-5} {\bibfield  {journal} {\bibinfo  {journal} {Commun. Math. Phys.}\ }\textbf {\bibinfo {volume} {279}},\ \bibinfo {pages} {251} (\bibinfo {year} {2008})}\BibitemShut {NoStop}%
\bibitem [{\citenamefont {Audenaert}\ \emph {et~al.}(2007)\citenamefont {Audenaert}, \citenamefont {Calsamiglia}, \citenamefont {Mu\~noz Tapia}, \citenamefont {Bagan}, \citenamefont {Masanes}, \citenamefont {Acin},\ and\ \citenamefont {Verstraete}}]{qchernoff2007}%
  \BibitemOpen
  \bibfield  {author} {\bibinfo {author} {\bibfnamefont {K.~M.~R.}\ \bibnamefont {Audenaert}}, \bibinfo {author} {\bibfnamefont {J.}~\bibnamefont {Calsamiglia}}, \bibinfo {author} {\bibfnamefont {R.}~\bibnamefont {Mu\~noz Tapia}}, \bibinfo {author} {\bibfnamefont {E.}~\bibnamefont {Bagan}}, \bibinfo {author} {\bibfnamefont {L.}~\bibnamefont {Masanes}}, \bibinfo {author} {\bibfnamefont {A.}~\bibnamefont {Acin}},\ and\ \bibinfo {author} {\bibfnamefont {F.}~\bibnamefont {Verstraete}},\ }\href {https://doi.org/10.1103/PhysRevLett.98.160501} {\bibfield  {journal} {\bibinfo  {journal} {Phys. Rev. Lett.}\ }\textbf {\bibinfo {volume} {98}},\ \bibinfo {pages} {160501} (\bibinfo {year} {2007})}\BibitemShut {NoStop}%
\bibitem [{\citenamefont {Chernoff}(1952)}]{chernoff1952}%
  \BibitemOpen
  \bibfield  {author} {\bibinfo {author} {\bibfnamefont {H.}~\bibnamefont {Chernoff}},\ }\href {https://doi.org/10.1214/aoms/1177729330} {\bibfield  {journal} {\bibinfo  {journal} {Ann. Math. Statist}\ }\textbf {\bibinfo {volume} {23}},\ \bibinfo {pages} {493} (\bibinfo {year} {1952})}\BibitemShut {NoStop}%
\bibitem [{\citenamefont {{Ogawa}}\ and\ \citenamefont {{Nagaoka}}(2000)}]{qsl2000}%
  \BibitemOpen
  \bibfield  {author} {\bibinfo {author} {\bibfnamefont {T.}~\bibnamefont {{Ogawa}}}\ and\ \bibinfo {author} {\bibfnamefont {H.}~\bibnamefont {{Nagaoka}}},\ }\href {https://doi.org/10.1109/18.887855} {\bibfield  {journal} {\bibinfo  {journal} {IEEE T. Inform. Theory}\ }\textbf {\bibinfo {volume} {46}},\ \bibinfo {pages} {2428} (\bibinfo {year} {2000})}\BibitemShut {NoStop}%
\bibitem [{\citenamefont {Hayashi}(2002)}]{Hayashi_2002}%
  \BibitemOpen
  \bibfield  {author} {\bibinfo {author} {\bibfnamefont {M.}~\bibnamefont {Hayashi}},\ }\href {https://doi.org/10.1088/0305-4470/35/50/307} {\bibfield  {journal} {\bibinfo  {journal} {J. Phys. A-Math. Gen.}\ }\textbf {\bibinfo {volume} {35}},\ \bibinfo {pages} {10759} (\bibinfo {year} {2002})}\BibitemShut {NoStop}%
\bibitem [{\citenamefont {Hayashi}(2007)}]{Hayashi_2007}%
  \BibitemOpen
  \bibfield  {author} {\bibinfo {author} {\bibfnamefont {M.}~\bibnamefont {Hayashi}},\ }\href {https://doi.org/10.1103/PhysRevA.76.062301} {\bibfield  {journal} {\bibinfo  {journal} {Phys. Rev. A}\ }\textbf {\bibinfo {volume} {76}},\ \bibinfo {pages} {062301} (\bibinfo {year} {2007})}\BibitemShut {NoStop}%
\bibitem [{\citenamefont {Brandão}\ and\ \citenamefont {Plenio}(2010)}]{qstein2010}%
  \BibitemOpen
  \bibfield  {author} {\bibinfo {author} {\bibfnamefont {F.~G. S.~L.}\ \bibnamefont {Brandão}}\ and\ \bibinfo {author} {\bibfnamefont {M.~B.}\ \bibnamefont {Plenio}},\ }\href {https://doi.org/10.1007/s00220-010-1005-z} {\bibfield  {journal} {\bibinfo  {journal} {Commun. Math. Phys.}\ }\textbf {\bibinfo {volume} {295}},\ \bibinfo {pages} {791} (\bibinfo {year} {2010})}\BibitemShut {NoStop}%
\bibitem [{\citenamefont {Li}(2014)}]{like2014}%
  \BibitemOpen
  \bibfield  {author} {\bibinfo {author} {\bibfnamefont {K.}~\bibnamefont {Li}},\ }\href {https://doi.org/10.1214/13-AOS1185} {\bibfield  {journal} {\bibinfo  {journal} {Ann. Statist.}\ }\textbf {\bibinfo {volume} {42}},\ \bibinfo {pages} {171} (\bibinfo {year} {2014})}\BibitemShut {NoStop}%
\bibitem [{\citenamefont {Chitambar}\ and\ \citenamefont {Gour}(2019)}]{qrt2019}%
  \BibitemOpen
  \bibfield  {author} {\bibinfo {author} {\bibfnamefont {E.}~\bibnamefont {Chitambar}}\ and\ \bibinfo {author} {\bibfnamefont {G.}~\bibnamefont {Gour}},\ }\href {https://doi.org/10.1103/RevModPhys.91.025001} {\bibfield  {journal} {\bibinfo  {journal} {Rev. Mod. Phys.}\ }\textbf {\bibinfo {volume} {91}},\ \bibinfo {pages} {025001} (\bibinfo {year} {2019})}\BibitemShut {NoStop}%
\bibitem [{\citenamefont {{Hayashi}}\ and\ \citenamefont {{Owari}}(2017)}]{restrictedqht2017}%
  \BibitemOpen
  \bibfield  {author} {\bibinfo {author} {\bibfnamefont {M.}~\bibnamefont {{Hayashi}}}\ and\ \bibinfo {author} {\bibfnamefont {M.}~\bibnamefont {{Owari}}},\ }\href {https://doi.org/10.1109/TIT.2017.2687932} {\bibfield  {journal} {\bibinfo  {journal} {IEEE T. Inform. Theory}\ }\textbf {\bibinfo {volume} {63}},\ \bibinfo {pages} {4008} (\bibinfo {year} {2017})}\BibitemShut {NoStop}%
\bibitem [{\citenamefont {{Brandão}}\ \emph {et~al.}(2020)\citenamefont {{Brandão}}, \citenamefont {{Harrow}}, \citenamefont {{Lee}},\ and\ \citenamefont {{Peres}}}]{restrictedqht2020}%
  \BibitemOpen
  \bibfield  {author} {\bibinfo {author} {\bibfnamefont {F.~G. S.~L.}\ \bibnamefont {{Brandão}}}, \bibinfo {author} {\bibfnamefont {A.~W.}\ \bibnamefont {{Harrow}}}, \bibinfo {author} {\bibfnamefont {J.~R.}\ \bibnamefont {{Lee}}},\ and\ \bibinfo {author} {\bibfnamefont {Y.}~\bibnamefont {{Peres}}},\ }\href {https://doi.org/10.1109/TIT.2020.2979704} {\bibfield  {journal} {\bibinfo  {journal} {IEEE T. Inform. Theory}\ }\textbf {\bibinfo {volume} {66}},\ \bibinfo {pages} {5037} (\bibinfo {year} {2020})}\BibitemShut {NoStop}%
\bibitem [{\citenamefont {Gour}\ and\ \citenamefont {Spekkens}(2008)}]{rfgour2008}%
  \BibitemOpen
  \bibfield  {author} {\bibinfo {author} {\bibfnamefont {G.}~\bibnamefont {Gour}}\ and\ \bibinfo {author} {\bibfnamefont {R.~W.}\ \bibnamefont {Spekkens}},\ }\href {https://doi.org/10.1088/1367-2630/10/3/033023} {\bibfield  {journal} {\bibinfo  {journal} {New J. Phys.}\ }\textbf {\bibinfo {volume} {10}},\ \bibinfo {pages} {033023} (\bibinfo {year} {2008})}\BibitemShut {NoStop}%
\bibitem [{\citenamefont {Piani}\ \emph {et~al.}(2016)\citenamefont {Piani}, \citenamefont {Cianciaruso}, \citenamefont {Bromley}, \citenamefont {Napoli}, \citenamefont {Johnston},\ and\ \citenamefont {Adesso}}]{asym2016}%
  \BibitemOpen
  \bibfield  {author} {\bibinfo {author} {\bibfnamefont {M.}~\bibnamefont {Piani}}, \bibinfo {author} {\bibfnamefont {M.}~\bibnamefont {Cianciaruso}}, \bibinfo {author} {\bibfnamefont {T.~R.}\ \bibnamefont {Bromley}}, \bibinfo {author} {\bibfnamefont {C.}~\bibnamefont {Napoli}}, \bibinfo {author} {\bibfnamefont {N.}~\bibnamefont {Johnston}},\ and\ \bibinfo {author} {\bibfnamefont {G.}~\bibnamefont {Adesso}},\ }\href {https://doi.org/10.1103/PhysRevA.93.042107} {\bibfield  {journal} {\bibinfo  {journal} {Phys. Rev. A}\ }\textbf {\bibinfo {volume} {93}},\ \bibinfo {pages} {042107} (\bibinfo {year} {2016})}\BibitemShut {NoStop}%
\bibitem [{\citenamefont {Mashhad}(2012)}]{Marvianphd}%
  \BibitemOpen
  \bibfield  {author} {\bibinfo {author} {\bibfnamefont {I.~M.}\ \bibnamefont {Mashhad}},\ }\href {http://hdl.handle.net/10012/7088} {\bibinfo {title} {Symmetry, asymmetry and quantum information}} (\bibinfo {year} {2012})\BibitemShut {NoStop}%
\bibitem [{\citenamefont {Marvian}\ and\ \citenamefont {Spekkens}(2014)}]{marviannc2014}%
  \BibitemOpen
  \bibfield  {author} {\bibinfo {author} {\bibfnamefont {I.}~\bibnamefont {Marvian}}\ and\ \bibinfo {author} {\bibfnamefont {R.~W.}\ \bibnamefont {Spekkens}},\ }\href {https://doi.org/10.1038/ncomms4821} {\bibfield  {journal} {\bibinfo  {journal} {Nature Communications}\ }\textbf {\bibinfo {volume} {5}},\ \bibinfo {pages} {3821} (\bibinfo {year} {2014})}\BibitemShut {NoStop}%
\bibitem [{\citenamefont {Bartlett}\ \emph {et~al.}(2007)\citenamefont {Bartlett}, \citenamefont {Rudolph},\ and\ \citenamefont {Spekkens}}]{ssr2007}%
  \BibitemOpen
  \bibfield  {author} {\bibinfo {author} {\bibfnamefont {S.~D.}\ \bibnamefont {Bartlett}}, \bibinfo {author} {\bibfnamefont {T.}~\bibnamefont {Rudolph}},\ and\ \bibinfo {author} {\bibfnamefont {R.~W.}\ \bibnamefont {Spekkens}},\ }\href {https://doi.org/10.1103/RevModPhys.79.555} {\bibfield  {journal} {\bibinfo  {journal} {Rev. Mod. Phys.}\ }\textbf {\bibinfo {volume} {79}},\ \bibinfo {pages} {555} (\bibinfo {year} {2007})}\BibitemShut {NoStop}%
\bibitem [{\citenamefont {Hiai}\ \emph {et~al.}(2009)\citenamefont {Hiai}, \citenamefont {Mosonyi},\ and\ \citenamefont {Hayashi}}]{qhtgs2009}%
  \BibitemOpen
  \bibfield  {author} {\bibinfo {author} {\bibfnamefont {F.}~\bibnamefont {Hiai}}, \bibinfo {author} {\bibfnamefont {M.}~\bibnamefont {Mosonyi}},\ and\ \bibinfo {author} {\bibfnamefont {M.}~\bibnamefont {Hayashi}},\ }\href {https://doi.org/10.1063/1.3234186} {\bibfield  {journal} {\bibinfo  {journal} {J. Math. Phys.}\ }\textbf {\bibinfo {volume} {50}},\ \bibinfo {pages} {103304} (\bibinfo {year} {2009})}\BibitemShut {NoStop}%
\bibitem [{\citenamefont {Dupuis}\ \emph {et~al.}(2012)\citenamefont {Dupuis}, \citenamefont {Kramer}, \citenamefont {Faist}, \citenamefont {Renes},\ and\ \citenamefont {Renner}}]{geentropy2012}%
  \BibitemOpen
  \bibfield  {author} {\bibinfo {author} {\bibfnamefont {F.}~\bibnamefont {Dupuis}}, \bibinfo {author} {\bibfnamefont {L.}~\bibnamefont {Kramer}}, \bibinfo {author} {\bibfnamefont {P.}~\bibnamefont {Faist}}, \bibinfo {author} {\bibfnamefont {J.~M.}\ \bibnamefont {Renes}},\ and\ \bibinfo {author} {\bibfnamefont {R.}~\bibnamefont {Renner}},\ }\bibinfo {title} {Generalized entropies},\ in\ \href {https://doi.org/10.1142/9789814449243_0008} {\emph {\bibinfo {booktitle} {XVIIth International Congress on Mathematical Physics}}}\ (\bibinfo  {publisher} {World Scientific},\ \bibinfo {year} {2012})\ pp.\ \bibinfo {pages} {134--153}\BibitemShut {NoStop}%
\bibitem [{\citenamefont {Buscemi}\ and\ \citenamefont {Gour}(2017)}]{lorenz}%
  \BibitemOpen
  \bibfield  {author} {\bibinfo {author} {\bibfnamefont {F.}~\bibnamefont {Buscemi}}\ and\ \bibinfo {author} {\bibfnamefont {G.}~\bibnamefont {Gour}},\ }\href {https://doi.org/10.1103/PhysRevA.95.012110} {\bibfield  {journal} {\bibinfo  {journal} {Phys. Rev. A}\ }\textbf {\bibinfo {volume} {95}},\ \bibinfo {pages} {012110} (\bibinfo {year} {2017})}\BibitemShut {NoStop}%
\bibitem [{\citenamefont {Wang}\ and\ \citenamefont {Renner}(2012)}]{qhtligong}%
  \BibitemOpen
  \bibfield  {author} {\bibinfo {author} {\bibfnamefont {L.}~\bibnamefont {Wang}}\ and\ \bibinfo {author} {\bibfnamefont {R.}~\bibnamefont {Renner}},\ }\href {https://doi.org/10.1103/PhysRevLett.108.200501} {\bibfield  {journal} {\bibinfo  {journal} {Phys. Rev. Lett.}\ }\textbf {\bibinfo {volume} {108}},\ \bibinfo {pages} {200501} (\bibinfo {year} {2012})}\BibitemShut {NoStop}%
\bibitem [{\citenamefont {Kullback}\ and\ \citenamefont {Leibler}(1951)}]{kullback1951}%
  \BibitemOpen
  \bibfield  {author} {\bibinfo {author} {\bibfnamefont {S.}~\bibnamefont {Kullback}}\ and\ \bibinfo {author} {\bibfnamefont {R.~A.}\ \bibnamefont {Leibler}},\ }\href {https://doi.org/10.1214/aoms/1177729694} {\bibfield  {journal} {\bibinfo  {journal} {Ann. Math. Statist.}\ }\textbf {\bibinfo {volume} {22}},\ \bibinfo {pages} {79} (\bibinfo {year} {1951})}\BibitemShut {NoStop}%
\bibitem [{\citenamefont {Umegaki}(1962)}]{umegaki1962}%
  \BibitemOpen
  \bibfield  {author} {\bibinfo {author} {\bibfnamefont {H.}~\bibnamefont {Umegaki}},\ }\href {https://doi.org/10.2996/kmj/1138844604} {\bibfield  {journal} {\bibinfo  {journal} {Kodai Math. Sem. Rep.}\ }\textbf {\bibinfo {volume} {14}},\ \bibinfo {pages} {59} (\bibinfo {year} {1962})}\BibitemShut {NoStop}%
\bibitem [{\citenamefont {Vedral}(2002)}]{reereview2002}%
  \BibitemOpen
  \bibfield  {author} {\bibinfo {author} {\bibfnamefont {V.}~\bibnamefont {Vedral}},\ }\href {https://doi.org/10.1103/RevModPhys.74.197} {\bibfield  {journal} {\bibinfo  {journal} {Rev. Mod. Phys.}\ }\textbf {\bibinfo {volume} {74}},\ \bibinfo {pages} {197} (\bibinfo {year} {2002})}\BibitemShut {NoStop}%
\bibitem [{\citenamefont {{Hayashi}}\ and\ \citenamefont {{Nagaoka}}(2003)}]{spectrumre2002}%
  \BibitemOpen
  \bibfield  {author} {\bibinfo {author} {\bibfnamefont {M.}~\bibnamefont {{Hayashi}}}\ and\ \bibinfo {author} {\bibfnamefont {H.}~\bibnamefont {{Nagaoka}}},\ }\href {https://doi.org/10.1109/TIT.2003.813556} {\bibfield  {journal} {\bibinfo  {journal} {IEEE T. Inform. Theory}\ }\textbf {\bibinfo {volume} {49}},\ \bibinfo {pages} {1753} (\bibinfo {year} {2003})}\BibitemShut {NoStop}%
\bibitem [{\citenamefont {{Datta}}(2009)}]{minmaxre2008}%
  \BibitemOpen
  \bibfield  {author} {\bibinfo {author} {\bibfnamefont {N.}~\bibnamefont {{Datta}}},\ }\href {https://doi.org/10.1109/TIT.2009.2018325} {\bibfield  {journal} {\bibinfo  {journal} {IEEE T. Inform. Theory}\ }\textbf {\bibinfo {volume} {55}},\ \bibinfo {pages} {2816} (\bibinfo {year} {2009})}\BibitemShut {NoStop}%
\bibitem [{\citenamefont {Korzekwa}\ \emph {et~al.}(2022)\citenamefont {Korzekwa}, \citenamefont {Puchała}, \citenamefont {Tomamichel},\ and\ \citenamefont {Życzkowski}}]{encoding2019}%
  \BibitemOpen
  \bibfield  {author} {\bibinfo {author} {\bibfnamefont {K.}~\bibnamefont {Korzekwa}}, \bibinfo {author} {\bibfnamefont {Z.}~\bibnamefont {Puchała}}, \bibinfo {author} {\bibfnamefont {M.}~\bibnamefont {Tomamichel}},\ and\ \bibinfo {author} {\bibfnamefont {K.}~\bibnamefont {Życzkowski}},\ }\href {https://doi.org/10.1109/TIT.2022.3157440} {\bibfield  {journal} {\bibinfo  {journal} {IEEE T. Inform. Theory}\ }\textbf {\bibinfo {volume} {68}},\ \bibinfo {pages} {4518} (\bibinfo {year} {2022})}\BibitemShut {NoStop}%
\bibitem [{\citenamefont {Gour}\ \emph {et~al.}(2009)\citenamefont {Gour}, \citenamefont {Marvian},\ and\ \citenamefont {Spekkens}}]{gour2009}%
  \BibitemOpen
  \bibfield  {author} {\bibinfo {author} {\bibfnamefont {G.}~\bibnamefont {Gour}}, \bibinfo {author} {\bibfnamefont {I.}~\bibnamefont {Marvian}},\ and\ \bibinfo {author} {\bibfnamefont {R.~W.}\ \bibnamefont {Spekkens}},\ }\href {https://doi.org/10.1103/PhysRevA.80.012307} {\bibfield  {journal} {\bibinfo  {journal} {Phys. Rev. A}\ }\textbf {\bibinfo {volume} {80}},\ \bibinfo {pages} {012307} (\bibinfo {year} {2009})}\BibitemShut {NoStop}%
\end{thebibliography}%


\end{document}